\documentclass[11pt,a4paper]{article}

\usepackage{amsmath,amssymb,amsthm}
\usepackage{amsfonts}
\usepackage{mathtools}
\usepackage{multirow}
\usepackage{placeins}
\usepackage{array}
\usepackage{tikz}
\usepackage[margin=2.5cm]{geometry}
\usepackage[colorlinks=true,linkcolor=blue,citecolor=red,urlcolor=blue]{hyperref}
\allowdisplaybreaks

\newtheorem{theorem}{Theorem}[section]
\newtheorem{corollary}[theorem]{Corollary}
\newtheorem{lemma}[theorem]{Lemma}
\newtheorem{proposition}[theorem]{Proposition}
\theoremstyle{definition}
\newtheorem{definition}[theorem]{Definition}
\newtheorem{example}[theorem]{Example}

\newcommand{\Z}{\mathbb{Z}}
\newcommand{\C}{{\mathcal C}}
\newcommand{\cH}{{\mathcal H}}
\newcommand{\cB}{{\mathcal B}}
\newcommand{\cP}{{\mathcal P}}
\newcommand{\cS}{{\mathcal S}}
\newcommand{\cF}{{\mathcal F}}
\newcommand{\zero}{{\mathbf{0}}}
\newcommand{\one}{{\mathbf{1}}}
\newcommand{\bp}{{\mathbf{p}}}
\newcommand{\uu}{\mathbf{u}}
\newcommand{\vv}{\mathbf{v}}
\newcommand{\ww}{\mathbf{w}}
\newcommand{\xx}{\mathbf{x}}
\newcommand{\yy}{\mathbf{y}}
\newcommand{\zz}{\mathbf{z}}
\newcommand{\ba}{\mathbf{a}}
\newcommand{\rank}{\operatorname{rank}}
\newcommand{\kernel}{\operatorname{ker}}
\newcommand{\K}{\operatorname{K}}
\newcommand{\bH}{\bar H}
\newcommand{\bcH}{\bar{\mathcal H}}
\newcommand{\bA}{\bar A}
\newcommand{\hvv}{\hat{\mathbf v}}
\newcommand{\hww}{\hat{\mathbf w}}
\newcommand{\huu}{\hat{\mathbf u}}

\begin{document}

\title{\textbf{Nonequivalence of $\mathbb{Z}_p\mathbb{Z}_{p^2}$-linear and $\mathbb{Z}_{p^s}$-linear
generalized Hadamard codes for odd primes}\thanks{This work has been granted by the Juan de la Cierva 2024 grant (JDC2024-053082-I), funded by Ministerio de Ciencia, Innovación y Universidades (MICIU/AEI/10.13039/501100011033) and co-funded by the European Social Fund Plus (FSE+).}}

\author{
Dipak K. Bhunia\\
\small Departament de Matem\`atiques,\\ \small Universitat Polit\`ecnica de Catalunya,\\
\small Barcelona, Spain
}
\date{}

\maketitle

\begin{abstract}
The $\mathbb{Z}_p\mathbb{Z}_{p^2}$-additive codes are subgroups of
$\mathbb{Z}_p^{\alpha_1}\times\mathbb{Z}_{p^2}^{\alpha_2}$, and they can be seen as linear
codes over $\mathbb{Z}_p$ when $\alpha_2=0$, as $\mathbb{Z}_{p^2}$-additive codes when
$\alpha_1=0$, and as $\mathbb{Z}_2\mathbb{Z}_4$-additive codes when $p=2$.  A
$\mathbb{Z}_p\mathbb{Z}_{p^2}$-linear generalized Hadamard (GH) code is a GH code over $\mathbb{Z}_p$
which is the Gray map image of a $\mathbb{Z}_p\mathbb{Z}_{p^2}$-additive code.  For every prime
$p$, there are recursive constructions of $\mathbb{Z}_p\mathbb{Z}_{p^2}$-linear GH codes
$H_p^{t_1,t_2}$ with $\alpha_1\neq0$ and $\alpha_2\neq0$, and of
$\mathbb{Z}_{p^s}$-linear GH codes $\bar H_p^{t_1,\dots,t_s}$.  We show that the kernel
blocks of a nonlinear code $H_p^{t_1,t_2}$ give equivalent punctured codes, and
that these punctured codes are linear.  Any of them is called the derived code
of $H_p^{t_1,t_2}$.  We also determine the ranks of the codes $H_p^{t_1,t_2}$
and $\bar H_p^{t_1,t_2}$.  For odd $p$, the rank of $H_p^{t_1,t_2}$ depends on the
Bernoulli numbers modulo $p$.  By comparing the derived codes and the ranks, we
show that, for odd $p$, no code $H_p^{t_1,t_2}$ is equivalent to any of these 
$\mathbb{Z}_{p^s}$-linear GH codes.  For $p=2$, a nonlinear code $H_2^{t_1,t_2}$ is
equivalent to a $\mathbb{Z}_{2^s}$-linear Hadamard code if and only if $t_2\geq2$.
Finally, we count the nonequivalent codes in the union of both families.
\end{abstract}



\section{Introduction}\label{sec:intro}

Let $\Z_{p^s}$ be the ring of integers modulo $p^s$ with $p$ prime and
$s\geq1$.  The set of $n$-tuples over $\Z_{p^s}$ is denoted by $\Z_{p^s}^n$,
and its elements are called vectors.  We denote the entries of a vector by the
same letter with subscripts, in normal font when the vector is written in bold.
For instance, $\xx=(x_1,\dots,x_n)$ and $y=(y_1,\dots,y_n)$.  A code over
$\Z_p$ of length $N$ is a nonempty subset of $\Z_p^N$, and it is linear if it
is a subspace of $\Z_p^N$.  The Hamming distance between two vectors of
$\Z_p^N$ is the number of coordinates in which they differ.  The minimum
distance of a code over $\Z_p$ is the smallest Hamming distance between two
different codewords.  A $\Z_{p^s}$-additive code is a subgroup of $\Z_{p^s}^n$,
and a $\Z_p\Z_{p^2}$-additive code is a subgroup of
$\Z_p^{\alpha_1}\times\Z_{p^2}^{\alpha_2}$.  A $\Z_p\Z_{p^2}$-additive code is
a linear code over $\Z_p$ when $\alpha_2=0$, a $\Z_{p^2}$-additive code when
$\alpha_1=0$, and a $\Z_2\Z_4$-additive code when $p=2$.  A $\Z_{p^s}$-additive
code $\C\subseteq\Z_{p^s}^n$ is isomorphic, as an abelian group, to
$$
\Z_{p^s}^{t_1}\times\Z_{p^{s-1}}^{t_2}\times\cdots\times\Z_p^{t_s}
$$
for unique nonnegative integers $t_1,\dots,t_s$, and we say that $\C$ is of
type $(n;t_1,\dots,t_s)$.  Similarly, a $\Z_p\Z_{p^2}$-additive code
$\C\subseteq\Z_p^{\alpha_1}\times\Z_{p^2}^{\alpha_2}$ is isomorphic to
$\Z_{p^2}^{t_1}\times\Z_p^{t_2}$ for unique nonnegative integers $t_1$ and
$t_2$, and we say that $\C$ is of type $(\alpha_1,\alpha_2;t_1,t_2)$.  Then,
$|\C|=p^{2t_1+t_2}$.  The Gray map from $\Z_{p^s}$ to
$\Z_p^{p^{s-1}}$ considered in
\cite{ZpsEquivalence,ZpZp2Construction,HadamardZps,ZpZp2Classification} is
recalled in Section \ref{sec:prelim}.  It is given in \cite{GrayIsometry}, and,
for $p=2$, it is Carlet's Gray map \cite{Carlet}.  It is also a particular case
of the Gray maps given in \cite{ShiKrotov2019}.  The Gray map transforms
$\Z_{p^s}$-additive codes and $\Z_p\Z_{p^2}$-additive codes into codes over
$\Z_p$, which are called $\Z_{p^s}$-linear codes and $\Z_p\Z_{p^2}$-linear
codes, respectively.  We say that they have the same type as the
corresponding additive codes.

We denote by $\zero$ and $\one$ the all-zero and the all-one vectors, and by
$e_k$ the vector whose $k$th entry is $1$ and whose other entries are $0$.  The
length of these vectors is always clear from the context.  Let $\cS_N$ be the
symmetric group of permutations on $\{1,\dots,N\}$, and let
$\Z_p^*=\Z_p\setminus\{0\}$.  For $\pi\in\cS_N$ and
$\ba=(a_1,\dots,a_N)\in(\Z_p^*)^N$, we define the monomial map $\mu_{\pi,\ba}$
of $\Z_p^N$ by $\mu_{\pi,\ba}(\xx)_{\pi(i)}=a_ix_i$ for $1\leq i\leq N$.  Thus,
the entry of $\xx$ at the position $i$ is multiplied by $a_i$ and moved to the
position $\pi(i)$.  When $\ba=\one$, we simply write $\pi$ for
$\mu_{\pi,\one}$.  Two codes $C_1$ and $C_2$ over $\Z_p$ of length $N$ are said
to be equivalent if there are a vector $\yy\in\Z_p^N$ and a monomial map $\mu$
such that
$$
C_2=\{\yy+\mu(\xx):\ \xx\in C_1\}.
$$
They are said to be permutation equivalent if moreover $\yy=\zero$ and $\mu$ is
a permutation.  In \cite{ZpZp2Construction,HadamardZps,ZpZp2Classification},
two codes are called equivalent when they differ by a monomial map, and in
\cite{ZpsEquivalence,EquivZ2s} when they differ by a translation and a
permutation.  Both notions are particular cases of the one above.  Therefore,
the nonequivalence results of this paper hold for each of them.

The rank of a code $C$ over $\Z_p$, denoted by $\rank(C)$, is the dimension
of its linear span $\langle C\rangle$.  Following \cite{BGH83,pKernel}, the
kernel of $C$ is
$$
\K(C)=\{\xx\in\Z_p^N:\ \xx+C=C\},
$$
and its dimension is denoted by $\kernel(C)$.  If $\zero\in C$, then $\K(C)$
is a linear subcode of $C$.  If $C$ is linear, then
$\K(C)=C=\langle C\rangle$.  Two equivalent codes containing $\zero$ have the
same rank and the same dimension of the kernel \cite{FullClass}.  Hence,
these parameters can be used to distinguish between nonequivalent codes.

A generalized Hadamard (GH) matrix $H(p,\lambda)=(h_{ij})$ of order $p\lambda$
over $\Z_p$ is a $p\lambda\times p\lambda$ matrix with entries from $\Z_p$ with
the following property \cite{jungnickel1979}.  For every $i$ and $j$ with
$1\leq i<j\leq p\lambda$, the multiset $\{h_{ik}-h_{jk}:1\leq k\leq p\lambda\}$
contains every element of $\Z_p$ exactly $\lambda$ times.  Two GH matrices are
equivalent if one of them can be obtained from the other by permuting rows and
columns, and by adding the same element of $\Z_p$ to all the entries of a row
or of a column.  Hence, every GH matrix is equivalent to a normalized one, that
is, to one whose first row and first column contain only zeros.  If $H$ is a
normalized GH matrix and $F_H$ is the set of its rows, then the code
$$
C_H=\bigcup_{\nu\in\Z_p}(F_H+\nu\one)
$$
is called a generalized Hadamard (GH) code over $\Z_p$
\cite{dougherty2015ranks}.  A GH code over $\Z_p$ of length $N$ has $pN$
codewords and minimum distance $(p-1)N/p$ \cite{dougherty2015ranks}.  Thus, any
two different codewords of a GH code over $\Z_p$ differ in at least a
proportion $(p-1)/p$ of the coordinates.  A $\Z_{p^s}$-linear code which is a
GH code is called a $\Z_{p^s}$-linear GH code, and similarly for
$\Z_p\Z_{p^2}$-linear codes.  For $p=2$, GH codes are the usual binary Hadamard
codes, and the $\Z_{2^s}$-linear and $\Z_2\Z_4$-linear GH codes are called
$\Z_{2^s}$-linear and $\Z_2\Z_4$-linear Hadamard codes.

For $p=2$ and $t\geq3$, the $\Z_2\Z_4$-linear Hadamard codes of length $2^t$
with $\alpha_1=0$ and with $\alpha_1\neq0$ are classified in
\cite{Kro:2001:Z4_Had_Perf,PRV06}.  There are $\lfloor(t-1)/2\rfloor$ and
$\lfloor t/2\rfloor$ nonequivalent such codes, respectively.  Moreover, every
$\Z_4$-linear Hadamard code is equivalent to a $\Z_2\Z_4$-linear Hadamard code
with $\alpha_1\neq0$ \cite{KV2015}.  Hence, for every even $t\geq4$, the codes
of exactly one of the $\lfloor t/2\rfloor$ classes of the latter codes are not
equivalent to any $\Z_4$-linear Hadamard code.  In \cite{EquivZ2s}, it is shown
that, for even $t$ with $4\leq t\leq10$, these codes are not equivalent to any
$\Z_{2^s}$-linear Hadamard code either.  The same statement for every even
$t\geq12$ is proposed there as further research.  The $\Z_{2^s}$-linear
Hadamard codes are also studied in \cite{KernelZ2s,fernandez2019mathbb}.  The
$\Z_2\Z_4\Z_8$-linear Hadamard codes are studied in
\cite{Z2Z4Z8Profile,OthersProfile,Z2Z4Z8Construction,Z2Z4Z8Linearity}.  In
\cite{OthersProfile}, it is shown that no nonlinear $\Z_2\Z_4\Z_8$-linear Hadamard
code is equivalent to a $\Z_4$-linear, a $\Z_2\Z_4$-linear or a
$\Z_{2^s}$-linear Hadamard code of the same length.

For any prime $p$, a recursive construction of $\Z_{p^s}$-linear GH codes
$\bH_p^{t_1,\dots,t_s}$ of length $p^t$, where $t+1=\sum_{i=1}^{s}(s-i+1)t_i$,
is given in \cite{HadamardZps,KernelZ2s}.  For $p=2$, every $\Z_{2^s}$-linear
Hadamard code is equivalent to one of these codes \cite{KernelZ2s,Krotov:2007}.
For odd $p$, it is not known whether the same holds.  As in
\cite{FullClass,ZpsEquivalence,HadamardZps}, by the $\Z_{p^s}$-linear GH codes
of length $p^t$ we mean the codes $\bH_p^{t_1,\dots,t_s}$ of length $p^t$, with
$s\geq1$.  Their linearity and their kernel are determined in
\cite{KernelZ2s,Kro:2001:Z4_Had_Perf,PRV06} for $p=2$, and in
\cite{HadamardZps} for odd $p$.  Some equivalences among them are given in
\cite{ZpsEquivalence,EquivZ2s}, and their full classification, for every prime
$p$, is obtained in \cite{FullClass}.

As in \cite{ZpZp2Construction,ZpZp2Classification}, we focus on the
$\Z_p\Z_{p^2}$-linear GH codes with $\alpha_1\neq0$ and $\alpha_2\neq0$, since
the codes with $\alpha_1=0$ are $\Z_{p^2}$-linear codes, and those with
$\alpha_2=0$ are linear codes over $\Z_p$.  A recursive construction of such
codes $H_p^{t_1,t_2}$ of type $(\alpha_1,\alpha_2;t_1,t_2)$, with
$t_1,t_2\geq1$, is given in \cite{ZpZp2Construction}.  Their length is $p^t$
with $t=2t_1+t_2-1$.  For $p=2$, they are the $\Z_2\Z_4$-linear Hadamard codes
with $\alpha_1\neq0$ given in \cite{PRV06,RSV08}.  It is also shown in
\cite{ZpZp2Construction} that several other recursive constructions give
permutation equivalent codes.  For odd $p$, it is proved in
\cite{ZpZp2Classification} that the codes $H_p^{t_1,t_2}$ are nonlinear, and
that the dimension of the kernel of $H_p^{t_1,t_2}$ is $t_1+t_2$.  Hence, there
are exactly $\lfloor t/2\rfloor$ nonequivalent such codes of length $p^t$.  It
is also shown in \cite{ZpZp2Classification} that, for odd $p$, the codes
$H_p^{1,t-1}$ and $H_p^{2,t-3}$ are not equivalent to any $\Z_{p^s}$-linear GH
code of the same length.  For $p=3$ and $2\leq t\leq8$, the values of the rank
and the dimension of the kernel show that no code $H_3^{t_1,t_2}$ of length
$3^t$ is equivalent to a $\Z_{3^s}$-linear GH code
\cite{ZpZp2Construction,ZpZp2Classification}.  Proving this for every odd prime
$p$ and every length was proposed in
\cite{ZpZp2Construction,ZpZp2Classification} as further research.

In this paper, we solve this problem, and we also study the case $p=2$.  Our
main tool is the derived code introduced in \cite{FullClass}.  Let $C$ be a
code over $\Z_p$ containing the all-zero vector.  Two coordinates of $C$ lie
in the same kernel block when the vectors of the kernel of $C$ which vanish
at one of them are exactly those which vanish at the other one.  When the
codes obtained by puncturing $C$ outside its kernel blocks are pairwise
equivalent, any of them is called the derived code of $C$.  It is an
equivalence invariant.  For $p=2$, the ranks of these punctured codes form
the kernel-block rank profile used in \cite{Z2Z4Z8Profile,OthersProfile}.  Every
nonlinear $\Z_{p^s}$-linear GH code is equivalent to a code
$\bH_p^{u_1,\dots,u_r}$ with $r\geq2$ and $u_1\geq2$ \cite{FullClass}.  For
odd $p$, the derived code of this code is
$\bH_p^{u_1,\dots,u_{r-2},u_{r-1}+1}$, and it is linear only when $r=2$.

We show that every kernel block of a nonlinear code $H_p^{t_1,t_2}$ has
$p^{t_1}$ coordinates, and that the derived code of $H_p^{t_1,t_2}$ is the
linear GH code $\bH_p^{t_1+1}$.  For odd $p$, it follows that, if
$H_p^{t_1,t_2}$ were equivalent to a $\Z_{p^s}$-linear GH code, then we would
have $t_1\geq2$, and $H_p^{t_1,t_2}$ would be equivalent to the
$\Z_{p^2}$-linear GH code $\bH_p^{t_1,t_2}$.  These two codes have the same
length and the same dimension of the kernel.  They also have the same number of
kernel blocks, of the same size, and the same derived code.  Figure
\ref{fig:blocks} shows their kernel blocks for $p=3$ and $(t_1,t_2)=(2,1)$.  In
order to separate them, we regard each coordinate of both codes as a function
of the coefficients which define the codewords, and we describe these functions
by polynomials over $\Z_p$.  In this way, we determine the ranks of
$H_p^{t_1,t_2}$ and $\bH_p^{t_1,t_2}$ for every prime $p$.  We show that the
rank of $\bH_p^{t_1,t_2}$ is smaller than the rank of $H_p^{t_1,t_2}$ for every
odd prime $p$.  Therefore, for odd $p$, the codes $H_p^{t_1,t_2}$ are not
equivalent to any $\Z_{p^s}$-linear GH code.  In particular, unlike the binary
case \cite{KV2015}, no $\Z_{p^2}$-linear GH code is equivalent to a code
$H_p^{t_1,t_2}$.

The rank of $H_p^{t_1,t_2}$ depends on the coefficients of the polynomial over
$\Z_p$, of degree at most $p-1$, whose value at the residue of
$x\in\{0,\dots,p-1\}$ is the residue of $(x^p-x)/p$ modulo $p$.  For odd $p$,
we show that the coefficients of the monomials of degree at least $2$ are given
by Bernoulli numbers modulo $p$.  Hence, the rank of $H_p^{t_1,t_2}$ depends on
the even integers $m$ with $2\leq m\leq p-3$ such that $p$ divides the
numerator of the Bernoulli number $B_m$ in lowest terms.  An odd prime $p$ is
regular if there is no such $m$.  For regular primes, we obtain a closed
formula for the rank of $H_p^{t_1,t_2}$.

For $p=2$, the same argument shows that, if a nonlinear code $H_2^{t_1,t_2}$
is equivalent to a $\Z_{2^s}$-linear Hadamard code, then it is equivalent to
one of two codes.  The first one is the $\Z_4$-linear Hadamard code
$\bH_2^{t_1+1,t_2-2}$, when $t_2\geq2$, and the second one is the
$\Z_8$-linear Hadamard code $\bH_2^{2,t_1-3,t_2}$, when $t_1\geq3$.  The second
code is excluded by the rank, and the first one is equivalent to
$H_2^{t_1,t_2}$ \cite{KV2015}.  Hence, a nonlinear code
$H_2^{t_1,t_2}$ is equivalent to a $\Z_{2^s}$-linear Hadamard code if and only
if $t_2\geq2$.  This answers the question of \cite{EquivZ2s} for every even
$t\geq4$.

Finally, we count the equivalence classes of the union of both families.  For
odd $p$ and $t\geq2$, there are exactly $\cP(t+1)-\cP(t)+\lfloor t/2\rfloor$
classes, where $\cP(n)$ is the number of partitions of $n$.  For $p=2$ and
$t\geq3$, there are exactly $\cP(t+1)-\cP(t)$ classes if $t$ is even, and
$\cP(t+1)-\cP(t)-1$ classes if $t$ is odd.  Figure \ref{fig:classes} shows
these classes for $t=8$.

The paper is organized as follows.  In Section \ref{sec:prelim}, we recall
the Gray map and the known results on the codes $\bH_p^{t_1,\dots,t_s}$ and
$H_p^{t_1,t_2}$.  In Section \ref{sec:blocks}, we recall the kernel blocks
and the derived code from \cite{FullClass}, and we determine the kernel
blocks and the derived code of the nonlinear codes $H_p^{t_1,t_2}$.  In
Section \ref{sec:rank}, we determine the ranks of $H_p^{t_1,t_2}$ and
$\bH_p^{t_1,t_2}$ for every prime $p$.  In Section \ref{sec:main}, we prove
the main result for odd $p$, and we count the nonequivalent codes.  In
Section \ref{sec:binary}, we study the case $p=2$.  Finally, in Section
\ref{sec:conclusions}, we give some conclusions and further research on this
topic.

\section{Preliminaries}\label{sec:prelim}

In this section, we fix the notation, and we recall the results on the codes
$\bH_p^{t_1,\dots,t_s}$ and $H_p^{t_1,t_2}$ that we will use.  Throughout the
paper, $p$ is a prime.  For two vectors $x=(x_1,\dots,x_n)$ and
$y=(y_1,\dots,y_n)$ over a ring, we write
$$
\langle x,y\rangle=\sum_{k=1}^{n}x_ky_k .
$$
The order of a vector $\uu$ in $\Z_{p^s}^n$, or in
$\Z_p^{\alpha_1}\times\Z_{p^2}^{\alpha_2}$, is the smallest positive integer
$k$ such that $k\uu=\zero$.  The invariance of the rank and the kernel
mentioned in Section \ref{sec:intro} is given by the following lemma.

\begin{lemma}[\cite{FullClass}]\label{lem:invariance}
Let $C,D\subseteq\Z_p^N$ be codes with $\zero\in C$ and $\zero\in D$, and
assume that $D=\yy+\mu(C)$ for some $\yy\in\Z_p^N$ and some monomial map
$\mu$.  Then, $\langle D\rangle=\mu(\langle C\rangle)$ and
$\K(D)=\mu(\K(C))$.  In particular, $C$ and $D$ have the same rank and the
same dimension of the kernel, and $C$ is linear if and only if $D$ is linear.
\end{lemma}

In several computations, the elements of $\Z_p$ are regarded as integers, or as
elements of $\Z_{p^2}$.  In order to avoid any ambiguity, we use the following
notation.  For an integer $w$, and also for $w\in\Z_{p^s}$, we denote by
$\overline{w}\in\Z_p$ the residue of $w$ modulo $p$.  For $a\in\Z_p$, we denote
by $\tilde a$ the unique integer in $\{0,\dots,p-1\}$ such that
$\overline{\tilde a}=a$.  We also regard $\tilde a$ as an element of
$\Z_{p^s}$.  In this way, $\{0,\dots,p-1\}$ is also regarded as a subset of
$\Z_{p^s}$.  Let $z$ be an element of an abelian group with $pz=0$, for
instance $z=p\in\Z_{p^2}$.  For $a\in\Z_p$, we write $az=za=\tilde az$.  Since
$pz=0$, we also have $az=bz$ for every integer $b$ with $\overline{b}=a$.  In
particular, for $a\in\Z_p$, we always regard $pa=p\tilde a$ as an element of
$\Z_{p^2}$, and $\overline{pa}=0$.  These notations are applied entrywise to
vectors and matrices.  For instance, if $\xi\in\Z_p^n$, then the entries of
$\tilde\xi\in\Z_{p^2}^n$ belong to $\{0,\dots,p-1\}$, and we have
$\overline{\tilde\xi}=\xi$ and $p\xi=p\tilde\xi$.

Every $u\in\Z_{p^s}$ can be written in a unique way as
$$
u=\sum_{i=0}^{s-1}\tilde u_i\,p^i
$$
with $u_0,\dots,u_{s-1}\in\Z_p$.  We write $u=[u_0,u_1,\dots,u_{s-1}]_p$, and
we call $u_i$ the digit of index $i$ of $u$.  Thus, $u_0=\overline{u}$.  The
Gray map $\phi_s:\Z_{p^s}\rightarrow\Z_p^{p^{s-1}}$ considered in
\cite{ZpsEquivalence,ZpZp2Construction,HadamardZps,ZpZp2Classification} is
$$
\phi_s(u)=(u_{s-1},u_{s-1},\dots,u_{s-1})+(u_0,\dots,u_{s-2})Y_{s-1},
$$
where $Y_{s-1}$ is a matrix of size $(s-1)\times p^{s-1}$ whose columns are all
the vectors of $\Z_p^{s-1}$, in a fixed order.  We index the $p^{s-1}$
coordinates of $\phi_s(u)$ by the columns of $Y_{s-1}$, that is, by the vectors
$h=(h_0,\dots,h_{s-2})\in\Z_p^{s-1}$.  With this convention, we have
\begin{equation}\label{eq:gray}
\phi_s(u)_h=u_{s-1}+\sum_{i=0}^{s-2}h_iu_i .
\end{equation}
For $s=1$, the vector $h$ is empty, and $\phi_1$ is the identity map of
$\Z_p$.  For $s=2$, we write $\phi=\phi_2$.  Then, $h\in\Z_p$, and
\eqref{eq:gray} becomes $\phi(u)_h=u_1+hu_0$.  We will use the following
property of $\phi$.

\begin{lemma}[\cite{HadamardZps}]\label{lem:shift}
Let $u\in\Z_{p^2}$ and $a\in\Z_p$.  Then, $\phi(u+pa)=\phi(u)+a\one$.
\end{lemma}

We denote by $\Phi_s:\Z_{p^s}^n\rightarrow\Z_p^{np^{s-1}}$ the
component-wise extension of $\phi_s$.  The $\Z_{p^s}$-linear code associated
with a $\Z_{p^s}$-additive code $\C\subseteq\Z_{p^s}^n$ is $\Phi_s(\C)$, and
it has length $np^{s-1}$.  We write the vectors of
$\Z_p^{\alpha_1}\times\Z_{p^2}^{\alpha_2}$ as $(u\mid u')$, with
$u\in\Z_p^{\alpha_1}$ and $u'\in\Z_{p^2}^{\alpha_2}$.  The set
$\Z_p^{\alpha_1}\times\Z_{p^2}^{\alpha_2}$ is a module over $\Z_{p^2}$, where
an element $\lambda\in\Z_{p^2}$ acts on the coordinates over $\Z_p$ through
its residue modulo $p$.  We denote by
$\Phi:\Z_p^{\alpha_1}\times\Z_{p^2}^{\alpha_2}\rightarrow
\Z_p^{\alpha_1+p\alpha_2}$ the map given by
$$
\Phi(u\mid u')=\bigl(u,\phi(u'_1),\dots,\phi(u'_{\alpha_2})\bigr),
$$
where $u'=(u'_1,\dots,u'_{\alpha_2})$.  The $\Z_p\Z_{p^2}$-linear code
associated with a $\Z_p\Z_{p^2}$-additive code
$\C\subseteq\Z_p^{\alpha_1}\times\Z_{p^2}^{\alpha_2}$ is $\Phi(\C)$, and it
has length $\alpha_1+p\alpha_2$.

We now recall the $\Z_{p^s}$-linear GH codes constructed in
\cite{HadamardZps,KernelZ2s}.  Following \cite{ZpZp2Classification}, we denote
them by $\bH_p^{t_1,\dots,t_s}$, since the notation $H_p^{t_1,t_2}$ will be
used for the $\Z_p\Z_{p^2}$-linear GH codes with $\alpha_1\neq0$ and
$\alpha_2\neq0$.  In \cite{FullClass}, they are denoted by
$H_p^{t_1,\dots,t_s}$.  The bar in $\bH_p^{t_1,\dots,t_s}$, and in the notation
$\bcH_p^{t_1,\dots,t_s}$ and $\bA_p^{t_1,\dots,t_s}$ below, is part of the
name, and it does not denote a residue modulo $p$.  Let $s\geq1$.  We put
$$
T_i=\bigl\{k\cdot p^{i-1}:\ 0\leq k\leq p^{s-i+1}-1\bigr\}=p^{i-1}\Z_{p^s}
$$
for $1\leq i\leq s$.  Let $t_1,\dots,t_s$ be nonnegative integers with
$t_1\geq1$.  Let $\bA_p^{t_1,\dots,t_s}$ be the matrix whose columns are
exactly, each of them once, all the vectors $\zz^T$ with
\begin{equation}\label{eq:columns}
\zz\in\{1\}\times T_1^{t_1-1}\times T_2^{t_2}\times\cdots\times T_s^{t_s}.
\end{equation}
The order of the columns only affects the order of the coordinates of the
resulting codes, and we fix one order.  The first $t_1$ rows of
$\bA_p^{t_1,\dots,t_s}$ have order $p^s$, the next $t_2$ rows have order
$p^{s-1}$, and so on.  The first row is $\one$.  Let $\bcH_p^{t_1,\dots,t_s}$
be the $\Z_{p^s}$-additive code generated by the rows of
$\bA_p^{t_1,\dots,t_s}$, and let
$\bH_p^{t_1,\dots,t_s}=\Phi_s(\bcH_p^{t_1,\dots,t_s})$.  Then,
$\bcH_p^{t_1,\dots,t_s}$ is of type $(n;t_1,\dots,t_s)$, where $n$ is the
number of columns of $\bA_p^{t_1,\dots,t_s}$ \cite{HadamardZps,KernelZ2s}.  We
call $(t_1,\dots,t_s)$ the parameters of $\bH_p^{t_1,\dots,t_s}$.

\begin{theorem}[\cite{HadamardZps,KernelZ2s}]\label{thm:GHZps}
Let $s\geq1$, and let $t_1,\dots,t_s$ be nonnegative integers with
$t_1\geq1$.  Then, $\bH_p^{t_1,\dots,t_s}$ is a GH code over $\Z_p$ of length
$p^t$, where
\begin{equation}\label{eq:length}
t+1=\sum_{i=1}^{s}(s-i+1)t_i .
\end{equation}
\end{theorem}

For $s=1$, we have $T_1=\Z_p$.  Hence, $\bH_p^{t+1}$ is the linear code over
$\Z_p$ of length $p^t$ generated by the matrix whose columns are all the
vectors of $\{1\}\times\Z_p^{t}$.  If we index its coordinates by the vectors
$v\in\Z_p^t$, then
\begin{equation}\label{eq:affine}
\bH_p^{t+1}=\bigl\{\bigl(\nu+\langle\gamma,v\rangle\bigr)_{v\in\Z_p^{t}}:\
\nu\in\Z_p,\ \gamma\in\Z_p^{t}\bigr\}.
\end{equation}
Thus, $\bH_p^{t+1}$ is formed by the evaluations of all the affine functions
on $\Z_p^t$.  It is the first order generalized Reed--Muller code of length
$p^t$, and it has dimension $t+1$ \cite{Key}.

\begin{theorem}[\cite{HadamardZps,KernelZ2s,Kro:2001:Z4_Had_Perf,PRV06}]
\label{thm:linZps}
Let $s\geq2$.  If $p$ is odd, then $\bH_p^{t_1,\dots,t_s}$ is linear if and
only if its parameters are $(1,0,\dots,0,t_s)$.  If $p=2$ and $s=2$, then
$\bH_2^{t_1,t_2}$ is linear if and only if $t_1\leq2$.  If $p=2$ and
$s\geq3$, then $\bH_2^{t_1,\dots,t_s}$ is linear if and only if its
parameters are $(1,0,\dots,0,t_s)$ or $(1,0,\dots,0,1,t_s)$.
\end{theorem}

\begin{definition}[\cite{FullClass}]\label{def:reduced}
A code $\bH_p^{t_1,\dots,t_s}$ with $s\geq2$ is called \emph{reduced} if
$t_1\geq2$.
\end{definition}

By Theorem \ref{thm:linZps}, every reduced code is nonlinear if $p$ is odd.
If $p=2$, a reduced code is linear if and only if $s=2$ and $t_1=2$.  The
dimension of the kernel of the nonlinear reduced codes is given by the
following result.

\begin{theorem}[\cite{HadamardZps,KernelZ2s,Kro:2001:Z4_Had_Perf,PRV06}]
\label{thm:kerZps}
Let $\bH=\bH_p^{t_1,\dots,t_s}$ be a nonlinear reduced code.  Then,
$\kernel(\bH)=t_1+\cdots+t_s$ if $p$ is odd, and
$\kernel(\bH)=t_1+\cdots+t_s+1$ if $p=2$.
\end{theorem}

The next result shows that, up to permutation equivalence, it is enough to
consider the reduced codes and the linear code $\bH_p^{t+1}$.

\begin{proposition}[\cite{FullClass}]\label{prop:reduction}
Let $\bH=\bH_p^{t_1,\dots,t_s}$ be a code of length $p^t$ with $s\geq2$.  If
$\bH$ is nonlinear, then it is permutation equivalent to a nonlinear reduced
code of length $p^t$.  If $\bH$ is linear, then it is permutation equivalent
to $\bH_p^{t+1}$.
\end{proposition}

A partition of a positive integer $n$ is a way of writing $n$ as a sum of
positive integers without regard to order, and $\cP(n)$ denotes the number
of partitions of $n$ \cite{Andrews}.  The following classification is the
main result of \cite{FullClass}.

\begin{theorem}[\cite{FullClass}]\label{thm:classZps}
Let $t\geq3$.  Then, the linear code $\bH_p^{t+1}$ and the nonlinear reduced
codes of length $p^t$ form a complete system of representatives for the
equivalence classes of the $\Z_{p^s}$-linear GH codes of length $p^t$, with
$s\geq1$.  The number of these classes is $\cP(t+1)-\cP(t)$ if $p$ is odd,
and $\cP(t+1)-\cP(t)-1$ if $p=2$.
\end{theorem}

We now recall the $\Z_p\Z_{p^2}$-linear GH codes with $\alpha_1\neq0$ and
$\alpha_2\neq0$ constructed in \cite{ZpZp2Construction}.  For an integer $i$,
we denote by $\mathbf{i}$ the vector whose entries are all equal to $i$.  Its
entries are regarded as elements of $\Z_p$ in the coordinates over $\Z_p$, and
as elements of $\Z_{p^2}$ in the coordinates over $\Z_{p^2}$.  Let
$$
A_p^{1,1}=\left(\begin{array}{cccc|cccc}
1&1&\cdots&1&p&p&\cdots&p\\
0&1&\cdots&p-1&1&2&\cdots&p-1
\end{array}\right).
$$
If $A$ is a generator matrix of a $\Z_p\Z_{p^2}$-additive code, we write
$A=(A_1\mid A_2)$, where $A_1$ is formed by the first $\alpha_1$ columns and
$A_2$ by the last $\alpha_2$ columns.  From a matrix
$A=A_p^{t_1,t_2}=(A_1\mid A_2)$ with $t_1,t_2\geq1$, we construct the matrix
\begin{equation}\label{eq:constr-p}
A_p^{t_1,t_2+1}=\left(\begin{array}{cccc|cccc}
A_1&A_1&\cdots&A_1&A_2&A_2&\cdots&A_2\\
\zero&\one&\cdots&\mathbf{p-1}&p\cdot\zero&p\cdot\one&\cdots&
p\cdot(\mathbf{p-1})
\end{array}\right).
\end{equation}
From a matrix $A=A_p^{t_1,1}=(A_1\mid A_2)$ with $t_1\geq1$, we construct the
matrix
\begin{equation}\label{eq:constr-p2}
A_p^{t_1+1,1}=\left(\begin{array}{cccc|ccccccc}
A_1&A_1&\cdots&A_1&pA_1&\cdots&pA_1&A_2&A_2&\cdots&A_2\\
\zero&\one&\cdots&\mathbf{p-1}&\one&\cdots&\mathbf{p-1}&\zero&\one&\cdots&
\mathbf{p^2-1}
\end{array}\right).
\end{equation}
In \eqref{eq:constr-p2}, $pA_1$ denotes the matrix $p\tilde A_1$ over
$\Z_{p^2}$.  As in \cite{ZpZp2Construction,ZpZp2Classification}, we obtain
$A_p^{t_1,t_2}$ from $A_p^{1,1}$ by adding first $t_1-1$ rows of order $p^2$
with \eqref{eq:constr-p2}, and then $t_2-1$ rows of order $p$ with
\eqref{eq:constr-p}.  The first row of $A_p^{t_1,t_2}$ is always
$(\one\mid\bp)$, and we denote it by $\vv_0$.  We denote by
$\vv_1,\dots,\vv_{t_1}$ the rows of order $p^2$, and by
$\ww_1,\dots,\ww_{t_2-1}$ the remaining rows of order $p$, in the order in
which they are added.  In particular, $\vv_1$ is the row which comes from the
second row of $A_p^{1,1}$.  Let $\cH_p^{t_1,t_2}$ be the
$\Z_p\Z_{p^2}$-additive code generated by $A_p^{t_1,t_2}$, and let
$H_p^{t_1,t_2}=\Phi(\cH_p^{t_1,t_2})$.  We call $(t_1,t_2)$ the parameters of
$H_p^{t_1,t_2}$.

\begin{example}\label{ex:matrices}
For $p=3$, the matrices $A_3^{1,1}$ and $A_3^{2,1}$ are
$$
A_3^{1,1}=\left(\begin{array}{ccc|cc}
1&1&1&3&3\\
0&1&2&1&2
\end{array}\right)
$$
and
$$
A_3^{2,1}=\left(\begin{array}{ccccccccc|cccccccccccc}
1&1&1&1&1&1&1&1&1&3&3&3&3&3&3&3&3&3&\cdots&3&3\\
0&1&2&0&1&2&0&1&2&0&3&6&0&3&6&1&2&1&\cdots&1&2\\
0&0&0&1&1&1&2&2&2&1&1&1&2&2&2&0&0&1&\cdots&8&8
\end{array}\right).
$$
The rows of $A_3^{2,1}$ are $\vv_0$, $\vv_1$ and $\vv_2$.  The code
$H_3^{2,1}$ is of type $(9,24;2,1)$, and it has length $9+3\cdot24=81$.
\end{example}

\begin{theorem}[\cite{ZpZp2Construction}]\label{thm:GHZpZp2}
Let $t_1,t_2\geq1$.  Then, $H_p^{t_1,t_2}$ is a $\Z_p\Z_{p^2}$-linear GH code
of type $(p^{t-t_1},p^{t-1}-p^{t-t_1-1};t_1,t_2)$ and length $p^t$, where
\begin{equation}\label{eq:t}
t=2t_1+t_2-1 .
\end{equation}
\end{theorem}

In particular, $\alpha_1=p^{t-t_1}$ and $\alpha_2=p^{t-1}-p^{t-t_1-1}$ are
nonzero.  The linearity and the kernel of the codes $H_p^{t_1,t_2}$ are given
by the following two results.

\begin{theorem}[\cite{ZpZp2Classification,PRV06}]\label{thm:linZpZp2}
Let $t_1,t_2\geq1$.  Then, $H_p^{t_1,t_2}$ is linear if and only if $p=2$ and
$t_1=1$.
\end{theorem}

\begin{theorem}[\cite{ZpZp2Classification,PRV06}]\label{thm:kerZpZp2}
Let $t_1,t_2\geq1$, and assume that $H_p^{t_1,t_2}$ is nonlinear.  Then,
$\K(H_p^{t_1,t_2})$ is the image under $\Phi$ of the subcode of
$\cH_p^{t_1,t_2}$ formed by the vectors of order at most $p$.  Moreover, the
vectors $\Phi(\vv_0)$, $\Phi(p\vv_k)$ with $1\leq k\leq t_1$, and
$\Phi(\ww_j)$ with $1\leq j\leq t_2-1$, form a basis of
$\K(H_p^{t_1,t_2})$.  In particular, $\kernel(H_p^{t_1,t_2})=t_1+t_2$.
\end{theorem}

By \eqref{eq:t}, the codes $H_p^{t_1,t_2}$ of length $p^t$ are the codes
$H_p^{t_1,t+1-2t_1}$ with $1\leq t_1\leq\lfloor t/2\rfloor$.  Their kernels
have dimension $t+1-t_1$, except for the linear code $H_2^{1,t-1}$, whose
kernel has dimension $t+1$.  Since these values are pairwise different,
Lemma \ref{lem:invariance} gives the following classification.

\begin{corollary}[\cite{ZpZp2Classification,PRV06}]\label{cor:classZpZp2}
Let $t\geq2$.  Then, the codes $H_p^{t_1,t_2}$ of length $p^t$ are the
$\lfloor t/2\rfloor$ codes $H_p^{t_1,t+1-2t_1}$, with
$1\leq t_1\leq\lfloor t/2\rfloor$, and they are pairwise nonequivalent.
\end{corollary}

Other recursive constructions of $\Z_p\Z_{p^2}$-linear GH codes with
$\alpha_1\neq0$ and $\alpha_2\neq0$ are given in \cite{ZpZp2Construction}.  In
these constructions, the entries $1,\dots,p-1$ of the second row of $A_p^{1,1}$
over $\Z_{p^2}$ are replaced by elements of $\Z_{p^2}$ whose residues modulo
$p$ are $1,\dots,p-1$ in some order.  The same can be done with the entries
$1,\dots,p-1$ of the last row of \eqref{eq:constr-p2} over the copies of
$pA_1$.  All these constructions give codes which are permutation equivalent to
the codes $H_p^{t_1,t_2}$ of the same type \cite{ZpZp2Construction}.
Therefore, the results of this paper also hold for those codes.

\section{Kernel blocks and derived codes}\label{sec:blocks}

In this section, we recall the kernel blocks and the derived code of a code
over $\Z_p$ from \cite{FullClass}, together with the derived codes of the
nonlinear reduced codes.  Then, we determine the kernel blocks and the
derived code of the nonlinear codes $H_p^{t_1,t_2}$.  As a consequence, we
show that, if such a code is equivalent to a $\Z_{p^s}$-linear GH code, then
it is equivalent to a nonlinear reduced code whose derived code is linear.

For $\xx\in\Z_p^N$ and $B\subseteq\{1,\dots,N\}$, we write $\xx|_B$ for the
vector obtained from $\xx$ by deleting the coordinates whose index lies
outside $B$.  We put
$$
D|_B=\{\xx|_B:\ \xx\in D\}
$$
for $D\subseteq\Z_p^N$, and we call it the code $D$ punctured outside $B$.  The
entries of $\xx|_B$ are indexed by the elements of $B$.  Two codes whose
coordinates are indexed by different sets of the same size are compared through
a bijection between these sets.  Another bijection only changes the result by a
permutation of the coordinates.  Hence, equivalence and permutation equivalence
make sense for such codes.  Since deleting coordinates is a linear map, every
$D\subseteq\Z_p^N$ satisfies
\begin{equation}\label{eq:span-puncture}
\langle D\rangle\big|_B=\bigl\langle D|_B\bigr\rangle .
\end{equation}

\begin{definition}[\cite{FullClass}]\label{def:blocks}
Let $C\subseteq\Z_p^N$ be a code with $\zero\in C$, and put $K=\K(C)$.  Let
$$
K_i=\{\xx\in K:\ x_i=0\}
$$
for $1\leq i\leq N$.  For $1\leq i,i'\leq N$, we write $i\sim_K i'$ if and only
if $K_i=K_{i'}$.  The classes of the equivalence relation $\sim_K$ on the set
of coordinates are called the \emph{kernel blocks} of $C$, and the set formed
by them is denoted by $\cB(C)$.
\end{definition}

When the all-one vector belongs to the kernel, the kernel blocks are
described by the values of the vectors of the kernel, as the next lemma
shows.

\begin{lemma}[\cite{FullClass}]\label{lem:one}
Let $C\subseteq\Z_p^N$ be a code with $\zero\in C$ and $\one\in\K(C)$.  Then,
two coordinates $i$ and $i'$ lie in the same kernel block of $C$ if and only
if $x_i=x_{i'}$ for every $\xx\in\K(C)$.
\end{lemma}

The next lemma shows that an equivalence carries kernel blocks to kernel
blocks, and punctured codes to equivalent punctured codes.

\begin{lemma}[\cite{FullClass}]\label{lem:blocks-invariance}
Let $C$ and $D$ be as in Lemma \ref{lem:invariance}, with $D=\yy+\mu(C)$ and
$\mu=\mu_{\pi,\ba}$.  Then, the map $B\mapsto\pi(B)$ is a bijection from
$\cB(C)$ onto $\cB(D)$.  Moreover, for every $B\in\cB(C)$, the codes $C|_B$
and $D|_{\pi(B)}$ are equivalent.
\end{lemma}

In general, the kernel blocks of a code do not give equivalent punctured
codes.  When they do, we obtain a new invariant.

\begin{definition}[\cite{FullClass}]\label{def:derived}
Let $C\subseteq\Z_p^N$ be a code with $\zero\in C$.  We say that $C$ is
\emph{block-homogeneous} if the codes $C|_B$, with $B$ running over
$\cB(C)$, are pairwise equivalent.  In that case, we write $\partial C$ for
any one of them, and we call it the \emph{derived code} of $C$.
\end{definition}

The derived code is well defined up to equivalence.  Thus, when we say that the
derived code of $C$ is a code $D$, we mean that $\partial C$ is equivalent to
$D$.  Moreover, the derived code is an equivalence invariant, as the following
result shows.

\begin{corollary}[\cite{FullClass}]\label{cor:derived-invariant}
Let $C$ and $D$ be equivalent codes over $\Z_p$, both containing $\zero$.  If
$C$ is block-homogeneous, then so is $D$.  Moreover, the two codes have the
same number of kernel blocks, all their kernel blocks have the same size, and
$\partial C$ and $\partial D$ are equivalent.
\end{corollary}

Following \cite{Z2Z4Z8Profile,OthersProfile}, we call the multiset of the dimensions
of the spaces $\langle C\rangle|_B$, with $B\in\cB(C)$, the
\emph{kernel-block rank profile} of $C$.  By \eqref{eq:span-puncture}, these
dimensions are the ranks of the punctured codes $C|_B$.  Hence, by Lemmas
\ref{lem:invariance} and \ref{lem:blocks-invariance}, equivalent codes
containing $\zero$ have the same kernel-block rank profile.  The kernel
blocks and the derived codes of the nonlinear reduced codes are given by the
following theorem.

\begin{theorem}[\cite{FullClass}]\label{thm:derivedZps}
Let $\bH=\bH_p^{t_1,\dots,t_s}$ be a nonlinear reduced code of length $p^t$,
and let $\kappa=\kernel(\bH)$.  Then, $\bH$ has $p^{\kappa-1}$ kernel blocks,
and each of them has $p^{\,t-\kappa+1}$ coordinates.  Moreover, $\bH$ is
block-homogeneous, and, for every $B\in\cB(\bH)$, the code $\bH|_B$ is
permutation equivalent to the $\Z_{p^{s-1}}$-linear GH code
\begin{equation}\label{eq:derivedZps}
\begin{cases}
\bH_p^{t_1,\dots,t_{s-2},\,t_{s-1}+1} & \text{if $p$ is odd},\\[1mm]
\bH_2^{t_1,\dots,t_{s-1}} & \text{if } p=2.
\end{cases}
\end{equation}
\end{theorem}

By Theorem \ref{thm:GHZps}, the code \eqref{eq:derivedZps} is a GH code over
$\Z_p$ of length $p^{\,t-\kappa+1}$.  If $p$ is odd, it is linear if and only
if $s=2$.  Indeed, for $s=2$, it is the linear code $\bH_p^{t_1+1}$.  For
$s\geq3$, it is a reduced code, since its first parameter is $t_1\geq2$, and
hence it is nonlinear by Theorem \ref{thm:linZps}.  If $p=2$, it is linear if
and only if $s=2$, or $s=3$ and $t_1=2$, by Theorem \ref{thm:linZps}.  A GH
code over $\Z_p$ of length $p^{\,t-\kappa+1}$ has $p^{\,t-\kappa+2}$ codewords.
Hence, the rank of \eqref{eq:derivedZps} is at least $t-\kappa+2$, with
equality if and only if this code is linear.  Since $\bH$ is block-homogeneous,
its kernel-block rank profile consists of the single value
$\rank(\partial\bH)$, repeated $p^{\kappa-1}$ times.  Moreover,
$\rank(\partial\bH)=t-\kappa+2$ if and only if $s=2$, or $p=2$, $s=3$ and
$t_1=2$.

We now study the codes $H_p^{t_1,t_2}$.  From now on in this section,
$t_1,t_2\geq1$.  By Theorem \ref{thm:linZpZp2}, the code $H_p^{t_1,t_2}$ is
nonlinear if and only if $p$ is odd or $t_1\geq2$.  We first describe the
columns of $A_p^{t_1,t_2}$ and the values of the codewords of
$H_p^{t_1,t_2}$.

Recall that, for a vector $x$ over $\Z_{p^2}$ and a vector $\xi$ over $\Z_p$,
the vectors $\overline{x}$, $\tilde\xi$ and $p\xi=p\tilde\xi$ are defined
entrywise as in Section \ref{sec:prelim}.  If $\xi\in\Z_p^n$ is nonzero, then
we denote by $i(\xi)$ the smallest index $i$ such that $\xi_i\neq0$.  We put
$$
X_n=\bigl\{x\in\Z_{p^2}^{n}:\ \overline{x}\neq\zero\ \text{and}\
x_{i(\overline{x})}\in\{1,\dots,p-1\}\bigr\}
$$
for $n\geq1$.  Thus, a vector $x$ belongs to $X_n$ if and only if some entry
of $x$ does not belong to $p\Z_{p^2}$, and the first such entry belongs to
$\{1,\dots,p-1\}\subseteq\Z_{p^2}$.  For instance, $X_1=\{1,\dots,p-1\}$.

\begin{lemma}\label{lem:X}
Let $\xi\in\Z_p^n$ be nonzero, and let $i=i(\xi)$.  Then, the vectors
$x\in X_n$ with $\overline{x}=\xi$ are exactly the vectors $\tilde\xi+py$ with
$y\in\Z_p^n$ and $y_i=0$.  In particular, there are $p^{n-1}$ such vectors,
and $|X_n|=(p^n-1)p^{n-1}$.
\end{lemma}

\begin{proof}
A vector $x\in\Z_{p^2}^n$ satisfies $\overline{x}=\xi$ if and only if
$x=\tilde\xi+py$ for some $y\in\Z_p^n$, and such a vector $y$ is unique.  In
that case, $i(\overline{x})=i$ and $x_i=\tilde\xi_i+py_i$.  Since
$\tilde\xi_i\in\{1,\dots,p-1\}$, the entry $x_i$ belongs to $\{1,\dots,p-1\}$
if and only if $y_i=0$.  This proves the first statement.  Hence, each of the
$p^n-1$ nonzero vectors $\xi$ is the residue modulo $p$ of exactly $p^{n-1}$
vectors of $X_n$, and the last statement follows.
\end{proof}

\begin{lemma}\label{lem:columns}
Let $A_p^{t_1,t_2}=(A_1\mid A_2)$.  Then, the columns of $A_1$ are exactly,
each of them once, the vectors $(1,a,c)^T$ with $a\in\Z_p^{t_1}$ and
$c\in\Z_p^{t_2-1}$.  Moreover, the columns of $A_2$ are exactly, each of them
once, the vectors $(p,x,pc)^T$ with $x\in X_{t_1}$ and $c\in\Z_p^{t_2-1}$.
\end{lemma}

\begin{proof}
We first consider the matrices $A_p^{t_1,1}$, and we use induction on $t_1$.
For $t_1=1$, the columns of $A_1$ are the vectors $(1,a)^T$ with $a\in\Z_p$,
and the columns of $A_2$ are the vectors $(p,x)^T$ with
$x\in\{1,\dots,p-1\}=X_1$.  Hence, the statement holds.  Assume that it holds
for $A_p^{t_1,1}=(A_1\mid A_2)$.  By \eqref{eq:constr-p2}, the first part of
$A_p^{t_1+1,1}$ is formed by $p$ copies of $A_1$, indexed by $\Z_p$, and the
new row takes the value $b$ on the copy with index $b\in\Z_p$.  Hence, the
columns of this part are the vectors $(1,a,b)^T$ with $a\in\Z_p^{t_1}$ and
$b\in\Z_p$, that is, the vectors $(1,a')^T$ with $a'\in\Z_p^{t_1+1}$.  The
second part is formed by $p-1$ copies of $pA_1$ and $p^2$ copies of $A_2$.  The
copy of $pA_1$ with index $d\in\{1,\dots,p-1\}$ gives the columns $(p,pa,d)^T$
with $a\in\Z_p^{t_1}$.  These are the vectors $(p,x')^T$ with $x'\in X_{t_1+1}$
and $i(\overline{x'})=t_1+1$.  The copy of $A_2$ with index $e\in\Z_{p^2}$
gives the columns $(p,x,e)^T$ with $x\in X_{t_1}$.  These are the vectors
$(p,x')^T$ with $x'\in X_{t_1+1}$ and $i(\overline{x'})\leq t_1$.  Hence, the
statement holds for $A_p^{t_1+1,1}$.

Finally, by \eqref{eq:constr-p}, the columns of the first part of
$A_p^{t_1,t_2+1}$ are the vectors $(1,a,c,c_{t_2})^T$, and the columns of the
second part are the vectors $(p,x,pc,pc_{t_2})^T$.  Here, $(1,a,c)^T$ and
$(p,x,pc)^T$ run over the columns of $A_p^{t_1,t_2}$, and $c_{t_2}$ runs over
$\Z_p$.  Hence, by induction on $t_2$, the statement holds for every $t_2\geq1$.
\end{proof}

In the vectors of Lemma \ref{lem:columns}, the first entry is the entry of
the row $\vv_0$, the next $t_1$ entries are those of the rows
$\vv_1,\dots,\vv_{t_1}$, and the last $t_2-1$ entries are those of the rows
$\ww_1,\dots,\ww_{t_2-1}$.  From now on, we write $X=X_{t_1}$.

By Lemma \ref{lem:columns}, we index the coordinates of $H_p^{t_1,t_2}$ as
follows.  The first $\alpha_1$ coordinates are indexed by the pairs $(a,c)$,
where $(1,a,c)^T$ is the corresponding column of $A_1$.  The last
$p\alpha_2$ coordinates are indexed by the triples $(x,c,h)$ with $x\in X$,
$c\in\Z_p^{t_2-1}$ and $h\in\Z_p$.  Here, $(p,x,pc)^T$ is the corresponding
column of $A_2$, and $h$ indexes the $p$ coordinates of $\phi$ as in Section
\ref{sec:prelim}.  When $t_2=1$, the vector $c$ is empty.  Every codeword of
$\cH_p^{t_1,t_2}$ can be written as
\begin{equation}\label{eq:u}
\uu=\varepsilon\vv_0+\sum_{k=1}^{t_1}\lambda_k\vv_k+
\sum_{j=1}^{t_2-1}\eta_j\ww_j ,
\end{equation}
where $\varepsilon\in\Z_p$, $\lambda=(\lambda_1,\dots,\lambda_{t_1})\in
\Z_{p^2}^{t_1}$ and $\eta=(\eta_1,\dots,\eta_{t_2-1})\in\Z_p^{t_2-1}$.  Since
the rows $\vv_0$ and $\ww_j$ have order $p$, the products
$\varepsilon\vv_0$ and $\eta_j\ww_j$ are defined as in Section
\ref{sec:prelim}.

\begin{lemma}\label{lem:values}
Let $\uu$ be as in \eqref{eq:u}.  Then,
$$
\Phi(\uu)_{(a,c)}=\varepsilon+\langle\overline{\lambda},a\rangle
+\langle\eta,c\rangle
$$
and
$$
\Phi(\uu)_{(x,c,h)}=\varepsilon+\langle\eta,c\rangle
+\phi\bigl(\langle\lambda,x\rangle\bigr)_h .
$$
\end{lemma}

\begin{proof}
By Lemma \ref{lem:columns}, the entry of $\uu$ at the coordinate $(a,c)$ is
$\varepsilon+\sum_{k}\overline{\lambda_k}a_k+\sum_{j}\eta_jc_j$, since
$\lambda_k$ acts on the coordinates over $\Z_p$ through its residue modulo $p$.
This proves the first equality, since $\Phi$ does not change the first
$\alpha_1$ coordinates.  The entry of $\uu$ at the coordinate of
$\cH_p^{t_1,t_2}$ with column $(p,x,pc)^T$ is
$$
p\varepsilon+\langle\lambda,x\rangle+\sum_{j=1}^{t_2-1}\eta_j\,pc_j
=\langle\lambda,x\rangle+p\bigl(\varepsilon+\langle\eta,c\rangle\bigr)
\in\Z_{p^2}.
$$
By Lemma \ref{lem:shift}, its image under $\phi$ is
$\phi(\langle\lambda,x\rangle)+(\varepsilon+\langle\eta,c\rangle)\one$.  This
proves the second equality.
\end{proof}

We can now determine the kernel blocks of a nonlinear code $H_p^{t_1,t_2}$.
We attach to each coordinate of $H_p^{t_1,t_2}$ the label
$$
\rho(a,c)=(\zero,c)\quad\text{and}\quad\rho(x,c,h)=(\overline{x},c),
$$
which belongs to $\Z_p^{t_1}\times\Z_p^{t_2-1}$.  Since $\overline{x}\neq\zero$
for every $x\in X$, the first $\alpha_1$ coordinates are exactly the
coordinates whose label has first component $\zero$.

\begin{theorem}\label{thm:blocks}
Let $H=H_p^{t_1,t_2}$ be nonlinear.  Then, two coordinates of $H$ lie in the
same kernel block if and only if they have the same label.  Moreover, every
element of $\Z_p^{t_1}\times\Z_p^{t_2-1}$ is the label of exactly $p^{t_1}$
coordinates.  Consequently, $H$ has $p^{t_1+t_2-1}$ kernel blocks, and each of
them has $p^{t_1}$ coordinates.
\end{theorem}

\begin{proof}
By using Lemma \ref{lem:values}, we evaluate the basis of $\K(H)$ given by
Theorem \ref{thm:kerZpZp2} at each coordinate.  The vector $\Phi(\vv_0)$
corresponds to $\varepsilon=1$, $\lambda=\zero$ and $\eta=\zero$, and hence
$\Phi(\vv_0)=\one$.  Let $1\leq k\leq t_1$.  The vector $\Phi(p\vv_k)$
corresponds to $\varepsilon=0$, $\lambda=pe_k$ and $\eta=\zero$.  Hence, it
takes the value $0$ at $(a,c)$, since $\overline{pe_k}=\zero$.  At $(x,c,h)$,
it takes the value $\phi(px_k)_h=\overline{x_k}$, because the digits of $px_k$
are $0$ and $\overline{x_k}$.  Let $1\leq j\leq t_2-1$.  The vector
$\Phi(\ww_j)$ takes the value $c_j$ at both $(a,c)$ and $(x,c,h)$.  Therefore,
at a coordinate with label $(\xi,c)$, the values of these $t_1+t_2$ basis
vectors form the vector $(1,\xi,c)$.

Since $\one\in\K(H)$, we can apply Lemma \ref{lem:one}.  For two coordinates
$i$ and $i'$, the condition $y_i=y_{i'}$ is linear in $\yy$.  Hence, it holds
for every $\yy\in\K(H)$ if and only if it holds for the vectors of a basis of
$\K(H)$.  Therefore, two coordinates lie in the same kernel block of $H$ if
and only if they have the same label.  For $c\in\Z_p^{t_2-1}$, the label
$(\zero,c)$ is the label of the $p^{t_1}$ coordinates $(a,c)$ with
$a\in\Z_p^{t_1}$.  Let $\xi\in\Z_p^{t_1}$ be nonzero.  By Lemma \ref{lem:X},
there are $p^{t_1-1}$ vectors $x\in X$ with $\overline{x}=\xi$.  Hence, the
label $(\xi,c)$ is the label of the $p^{t_1}$ coordinates $(x,c,h)$ with
$\overline{x}=\xi$ and $h\in\Z_p$.  Therefore, the kernel blocks of $H$ are in
bijection with $\Z_p^{t_1}\times\Z_p^{t_2-1}$, and each of them has $p^{t_1}$
coordinates.
\end{proof}

Let $H=H_p^{t_1,t_2}$ be nonlinear.  For
$(\xi,c)\in\Z_p^{t_1}\times\Z_p^{t_2-1}$, we denote by $B_{(\xi,c)}$ the set of
coordinates of $H$ whose label is $(\xi,c)$.  By the definition of $\rho$, we
have
$$
B_{(\xi,c)}=\begin{cases}
\bigl\{(a,c):\ a\in\Z_p^{t_1}\bigr\} & \text{if } \xi=\zero,\\[1mm]
\bigl\{(x,c,h):\ x\in X,\ \overline{x}=\xi,\ h\in\Z_p\bigr\}
& \text{if } \xi\neq\zero.
\end{cases}
$$
By Theorem \ref{thm:blocks}, the kernel blocks of $H$ are exactly the sets
$B_{(\xi,c)}$.  We call $B_{(\xi,c)}$ the kernel block of $H$ with
\emph{label} $(\xi,c)$.

By Theorem \ref{thm:kerZpZp2}, the kernel of a nonlinear code
$H_p^{t_1,t_2}$ has dimension $\kappa=t_1+t_2$, and \eqref{eq:t} gives
$t-\kappa+1=t_1$.  Hence, by Theorem \ref{thm:blocks}, the code
$H_p^{t_1,t_2}$ has $p^{\kappa-1}$ kernel blocks, each of them with
$p^{\,t-\kappa+1}$ coordinates.  By Theorem \ref{thm:derivedZps}, a nonlinear
reduced code of length $p^t$ with kernel dimension $\kappa$ has the same
number of kernel blocks, of the same size.  The kernel blocks of
$H_p^{t_1,t_2}$ also give equivalent punctured codes, as we show now.

\begin{theorem}\label{thm:derived}
Let $H=H_p^{t_1,t_2}$ be nonlinear.  Then, for every kernel block $B$ of $H$,
the code $H|_B$ is permutation equivalent to the linear GH code
$\bH_p^{t_1+1}$ of length $p^{t_1}$.  Consequently, $H$ is block-homogeneous,
and its derived code $\partial H$ is, up to equivalence, the linear code
$\bH_p^{t_1+1}$.
\end{theorem}

\begin{proof}
Let $R=\bH_p^{t_1+1}$, with its coordinates indexed by the vectors
$v\in\Z_p^{t_1}$ as in \eqref{eq:affine}.  For every kernel block $B$, we
give a bijection $\beta_B$ from $B$ onto $\Z_p^{t_1}$ which carries $H|_B$
onto $R$.  We write the codewords of $H$ as $\Phi(\uu)$, with $\uu$ as in
\eqref{eq:u}.

First, let $B=B_{(\zero,c)}$, and let $\beta_B(a,c)=a$.  By Lemma
\ref{lem:values}, the codeword $\Phi(\uu)$ takes the value
$\nu+\langle\overline{\lambda},a\rangle$ at $(a,c)$, where
$\nu=\varepsilon+\langle\eta,c\rangle$.  For fixed $\lambda$ and $\eta$, the
element $\nu$ runs over $\Z_p$ as $\varepsilon$ runs over $\Z_p$.  Moreover,
$\overline{\lambda}$ runs over $\Z_p^{t_1}$ as $\lambda$ runs over
$\Z_{p^2}^{t_1}$.  Hence, by \eqref{eq:affine}, the map $\beta_B$ carries
$H|_B$ onto $R$.

Now, let $B=B_{(\xi,c)}$ with $\xi\neq\zero$, and put $i=i(\xi)$.  By Lemma
\ref{lem:X}, the vectors $x\in X$ with $\overline{x}=\xi$ are exactly the
vectors $\tilde\xi+py$ with $y\in P_i$, where
$$
P_i=\{y\in\Z_p^{t_1}:\ y_i=0\}.
$$
Let $\psi:P_i\times\Z_p\rightarrow\Z_p^{t_1}$ be the map given by
$\psi(y,h)=y+h\xi$.  Since $\xi_i\neq0$, the map $\psi$ is a bijection.  Its
inverse maps $v$ to $(v-h\xi,h)$, where $h=v_i\xi_i^{-1}$.  Let
$\beta_B(\tilde\xi+py,c,h)=\psi(y,h)$.  Put
$\omega=\langle\lambda,\tilde\xi\rangle\in\Z_{p^2}$, and let
$[\omega_0,\omega_1]_p$ be its $p$-ary expansion.  Since
$\langle\lambda,\tilde\xi+py\rangle=\omega+p\langle\overline{\lambda},y\rangle$,
Lemma \ref{lem:shift} and \eqref{eq:gray} give
$$
\phi\bigl(\langle\lambda,\tilde\xi+py\rangle\bigr)_h
=\omega_1+\langle\overline{\lambda},y\rangle+h\,\omega_0 .
$$
Moreover, $\omega_0=\overline{\omega}=\langle\overline{\lambda},\xi\rangle$,
since $\overline{\tilde\xi}=\xi$.  Therefore, by Lemma \ref{lem:values}, the
codeword $\Phi(\uu)$ takes the value
$$
\nu+\langle\overline{\lambda},y\rangle+h\langle\overline{\lambda},\xi\rangle
=\nu+\bigl\langle\overline{\lambda},\psi(y,h)\bigr\rangle
$$
at $(\tilde\xi+py,c,h)$, where $\nu=\varepsilon+\langle\eta,c\rangle+\omega_1$.
As before, $\nu$ runs over $\Z_p$ as $\varepsilon$ runs over $\Z_p$, and
$\overline{\lambda}$ runs over $\Z_p^{t_1}$.  Hence, by \eqref{eq:affine},
$\beta_B$ also carries $H|_B$ onto $R$ in this case.

Therefore, all the codes $H|_B$ are permutation equivalent to $R$, and hence
to each other.  Thus, $H$ is block-homogeneous, and $\partial H$ is the
linear code $R$, up to equivalence.
\end{proof}

\begin{corollary}\label{cor:profile}
Let $H=H_p^{t_1,t_2}$ be nonlinear.  Then, the kernel-block rank profile of
$H$ consists of the single value $t_1+1$, repeated $p^{t_1+t_2-1}$ times.
\end{corollary}

\begin{proof}
By Theorem \ref{thm:derived}, every punctured code $H|_B$ is permutation
equivalent to $\bH_p^{t_1+1}$.  Hence, it is a linear code of dimension
$t_1+1$.  By \eqref{eq:span-puncture}, it follows that every space
$\langle H\rangle|_B$ has dimension $t_1+1$.  By Theorem \ref{thm:blocks},
there are $p^{t_1+t_2-1}$ kernel blocks.
\end{proof}

From now on, when a code $H_p^{t_1,t_2}$ is compared with a $\Z_{p^s}$-linear
GH code, we denote the latter by $\bH_p^{u_1,\dots,u_r}$, in order to
distinguish the parameters of the two codes.  Then, the results on the codes
$\bH_p^{t_1,\dots,t_s}$ recalled above apply with $r$ instead of $s$ and $u_i$
instead of $t_i$.  With $\kappa=t_1+t_2$ as above, the value $t_1+1$ is equal
to $t-\kappa+2$.  By the discussion after Theorem \ref{thm:derivedZps}, a
nonlinear reduced code $\bH_p^{u_1,\dots,u_r}$ of length $p^t$ with kernel
dimension $\kappa$ has this kernel-block rank profile if and only if $r=2$, or
$p=2$, $r=3$ and $u_1=2$.  For $p=2$, the number and the size of the kernel
blocks of $H_2^{t_1,t_2}$, and its kernel-block rank profile, are also given in
\cite{OthersProfile}.

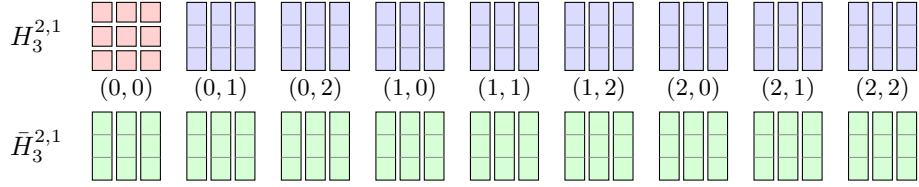
\begin{figure}[!ht]
\centering
\begin{tikzpicture}[font=\footnotesize]
\node[anchor=east] at (-0.25,0.45) {\small $H_3^{2,1}$};
\foreach \i in {0,1,2}{
  \foreach \j in {0,1,2}{
    \filldraw[fill=red!18,draw=black,line width=0.4pt]
      (\i*0.32,\j*0.32) rectangle ++(0.26,0.26);}}
\node at (0.45,-0.25) {$(0,0)$};
\foreach \b/\lab in {1/{(0,1)},2/{(0,2)},3/{(1,0)},4/{(1,1)},5/{(1,2)},
                     6/{(2,0)},7/{(2,1)},8/{(2,2)}}{
  \begin{scope}[xshift=\b*1.25cm]
    \foreach \i in {0,1,2}{
      \filldraw[fill=blue!14,draw=black,line width=0.4pt]
        (\i*0.32,0) rectangle ++(0.26,0.9);
      \draw[black!45,line width=0.2pt] (\i*0.32,0.3)--++(0.26,0);
      \draw[black!45,line width=0.2pt] (\i*0.32,0.6)--++(0.26,0);}
    \node at (0.45,-0.25) {$\lab$};
  \end{scope}}
\begin{scope}[yshift=-1.45cm]
\node[anchor=east] at (-0.25,0.45) {\small $\bH_3^{2,1}$};
\foreach \b in {0,...,8}{
  \begin{scope}[xshift=\b*1.25cm]
    \foreach \i in {0,1,2}{
      \filldraw[fill=green!16,draw=black,line width=0.4pt]
        (\i*0.32,0) rectangle ++(0.26,0.9);
      \draw[black!45,line width=0.2pt] (\i*0.32,0.3)--++(0.26,0);
      \draw[black!45,line width=0.2pt] (\i*0.32,0.6)--++(0.26,0);}
  \end{scope}}
\end{scope}
\end{tikzpicture}
\caption{The kernel blocks of $H_3^{2,1}$, with their labels, and of
$\bH_3^{2,1}$.  A single cell is a coordinate over $\Z_3$, and a column of
three cells comes from one coordinate over $\Z_9$.}
\label{fig:blocks}
\end{figure}

\begin{example}\label{ex:blocks}
Let $p=3$ and $(t_1,t_2)=(2,1)$, and let $H=H_3^{2,1}$ be the code of Example
\ref{ex:matrices}.  Then, $t=4$, and the vector $c$ is empty.  Hence, the
coordinates of $H$ are indexed by the vectors $a\in\Z_3^2$ and by the pairs
$(x,h)$ with $x\in X$ and $h\in\Z_3$, and the labels are the vectors
$\xi\in\Z_3^2$.  The set $X=X_2$ is formed by the $18$ vectors $(1,x_2)$ and
$(2,x_2)$ with $x_2\in\Z_9$, and by the $6$ vectors $(x_1,1)$ and $(x_1,2)$
with $x_1\in\{0,3,6\}$.  These are the columns of the second part of
$A_3^{2,1}$ without their first entry.  Since $c$ is empty, we write $B_\xi$
for the kernel block with label $\xi$.  By Theorem \ref{thm:blocks}, $H$ has
the $9$ kernel blocks $B_\xi$ with $\xi\in\Z_3^2$, and each of them has $9$
coordinates.  For instance, $B_{(0,0)}$ is formed by the $9$ coordinates of the
first part, $B_{(0,1)}$ is formed by the coordinates $((x_1,1),h)$ with
$x_1\in\{0,3,6\}$ and $h\in\Z_3$, and $B_{(2,1)}$ is formed by the coordinates
$((2,x_2),h)$ with $x_2\in\{1,4,7\}$ and $h\in\Z_3$.  The first row of Figure
\ref{fig:blocks} shows all these blocks.

Consider the block $B_{(2,1)}$.  Here, $i(\xi)=1$ and $\tilde\xi=(2,1)$, and the
vectors $x\in X$ with $\overline{x}=(2,1)$ are the vectors $(2,1+3y_2)$ with
$y_2\in\Z_3$.  Let $\uu=\varepsilon\vv_0+\lambda_1\vv_1+\lambda_2\vv_2$, and
put $\omega=2\lambda_1+\lambda_2\in\Z_9$, with $3$-ary expansion
$[\omega_0,\omega_1]_3$.  By Lemmas \ref{lem:values} and \ref{lem:shift}, the
codeword $\Phi(\uu)$ takes the value
$$
\varepsilon+\phi\bigl(\omega+3\overline{\lambda_2}\,y_2\bigr)_h
=(\varepsilon+\omega_1)+\overline{\lambda_2}\,y_2
+h\bigl(2\overline{\lambda_1}+\overline{\lambda_2}\bigr)
$$
at the coordinate $((2,1+3y_2),h)$.  With $v=(2h,y_2+h)$, this value is
$(\varepsilon+\omega_1)+\langle\overline{\lambda},v\rangle$, and $v$ runs over
$\Z_3^2$ as $(y_2,h)$ runs over $\Z_3^2$.  Hence, $H|_{B_{(2,1)}}$ is formed
by the evaluations of all the affine functions on $\Z_3^2$.  By Theorem
\ref{thm:derived}, the same holds for every kernel block.  Therefore, the
derived code of $H$ is the linear code $\bH_3^{3}$ of length $9$ and
dimension $3$.

The $\Z_9$-linear GH code $\bH_3^{2,1}$ has the same length $3^4$, by Theorem
\ref{thm:GHZps}, and the same kernel dimension $3$, by Theorem
\ref{thm:kerZps}.  By Theorem \ref{thm:derivedZps}, it also has $9$ kernel
blocks of size $9$, and its derived code is $\bH_3^{3}$.  The columns of
$\bA_3^{2,1}$ are the vectors $(1,z,3b)^T$ with $z\in\Z_9$ and $b\in\Z_3$.  Let
$\mathbf r_1$, $\mathbf r_2$ and $\mathbf r_3$ be the rows of $\bA_3^{2,1}$.
Every entry of $3\mathbf r_1$, $3\mathbf r_2$ and $\mathbf r_3$ belongs to
$3\Z_9$.  Hence, by Lemma \ref{lem:shift}, adding one of these vectors to a
codeword of $\bcH_3^{2,1}$ adds its image under $\Phi_2$ to the corresponding
codeword of $\bH_3^{2,1}$.  Therefore, the vectors $\Phi_2(3\mathbf r_1)$,
$\Phi_2(3\mathbf r_2)$ and $\Phi_2(\mathbf r_3)$ belong to the kernel of
$\bH_3^{2,1}$.  At the coordinates which come from the column $(1,z,3b)^T$,
they take the values $1$, $\overline{z}$ and $b$, respectively.  Thus, they are
linearly independent, and they form a basis of the kernel.  By Lemma
\ref{lem:one}, it follows that two coordinates of $\bH_3^{2,1}$ lie in the same
kernel block if and only if their columns have the same values of $b$ and
$\overline{z}$.  These blocks are shown in the second row of
Figure \ref{fig:blocks}.  Hence, the kernel blocks do not distinguish
$H_3^{2,1}$ from $\bH_3^{2,1}$.  We come back to this example in Section
\ref{sec:rank}.
\end{example}

We end this section with a consequence of Theorem \ref{thm:derived}, which
will be used in Sections \ref{sec:main} and \ref{sec:binary}.

\begin{proposition}\label{prop:candidates}
Let $H=H_p^{t_1,t_2}$ be nonlinear, and let $p^t$ be its length.  Assume that
$H$ is equivalent to a $\Z_{p^s}$-linear GH code, for some $s\geq1$.  Then,
$H$ is equivalent to a nonlinear reduced code $\bH_p^{u_1,\dots,u_r}$ of
length $p^t$ whose derived code is linear and has length $p^{t_1}$.
\end{proposition}

\begin{proof}
Let $\bH$ be a $\Z_{p^s}$-linear GH code equivalent to $H$.  Since $H$ is
nonlinear, $\bH$ is nonlinear by Lemma \ref{lem:invariance}.  In particular,
$s\geq2$, since the codes with $s=1$ are linear.  By Proposition
\ref{prop:reduction}, the code $\bH$ is permutation equivalent to a nonlinear
reduced code $\bH'=\bH_p^{u_1,\dots,u_r}$ of length $p^t$.  Hence, $H$ is
equivalent to $\bH'$.  By Theorem \ref{thm:derived}, the code $H$ is
block-homogeneous, and $\partial H$ is equivalent to the linear code
$\bH_p^{t_1+1}$ of length $p^{t_1}$.  By Corollary
\ref{cor:derived-invariant}, the code $\bH'$ is also block-homogeneous, and
$\partial\bH'$ is equivalent to $\partial H$.  Therefore, by Lemma
\ref{lem:invariance}, the derived code of $\bH'$ is linear, and its length is
$p^{t_1}$.
\end{proof}

\section{The ranks of \texorpdfstring{$H_p^{t_1,t_2}$ and $\bH_p^{t_1,t_2}$}{the two codes}}\label{sec:rank}

In this section, we determine the ranks of the $\Z_p\Z_{p^2}$-linear GH codes
$H_p^{t_1,t_2}$, with $t_1,t_2\geq1$, and of the $\Z_{p^2}$-linear GH codes
$\bH_p^{t_1,t_2}$, with $t_1\geq1$ and $t_2\geq0$.  For $t_2\geq1$, the codes
$H_p^{t_1,t_2}$ and $\bH_p^{t_1,t_2}$ have the same length $p^t$, with
$t=2t_1+t_2-1$, by \eqref{eq:length} and \eqref{eq:t}.  If moreover $p$ is
odd and $t_1\geq2$, then they also have the same dimension of the kernel and
the same derived code, by Theorems \ref{thm:kerZps}, \ref{thm:kerZpZp2},
\ref{thm:derivedZps} and \ref{thm:derived}.  Hence, another invariant is
needed to distinguish them.

For a finite set $M$, we denote by $\Z_p^M$ the vector space over $\Z_p$
formed by the functions from $M$ to $\Z_p$.  The next lemma computes the
rank of a code from the functions given by its coordinates.

\begin{lemma}\label{lem:evaluation}
Let $M$ be a finite set, and let $E:M\rightarrow\Z_p^N$ be a map.  For
$1\leq j\leq N$, let $f_j\in\Z_p^M$ be the function given by
$f_j(m)=E(m)_j$.  Then, the dimension of $\langle E(M)\rangle$ is equal to the
dimension of the subspace of $\Z_p^M$ spanned by $f_1,\dots,f_N$.
\end{lemma}

\begin{proof}
Let $m_1,\dots,m_n$ be the elements of $M$, and let $G$ be the $n\times N$
matrix whose $i$th row is $E(m_i)$.  The row space of $G$ is
$\langle E(M)\rangle$.  The map which sends a function $f\in\Z_p^M$ to the
vector $(f(m_1),\dots,f(m_n))^T$ is an isomorphism of vector spaces from
$\Z_p^M$ onto $\Z_p^n$.  Since $f_j(m_i)=E(m_i)_j$, it sends $f_j$ to the
$j$th column of $G$.  Hence, it maps the subspace spanned by $f_1,\dots,f_N$
onto the column space of $G$, and these two spaces have the same dimension.
Since the row rank and the column rank of a matrix coincide, the statement
follows.
\end{proof}

We call $f_1,\dots,f_N$ the coordinate functions of $E$.  We now describe the
coordinate functions of the two codes.  For $H_p^{t_1,t_2}$, let $M$ be the
set of triples $(\varepsilon,\lambda,\eta)$ with $\varepsilon\in\Z_p$,
$\lambda\in\Z_{p^2}^{t_1}$ and $\eta\in\Z_p^{t_2-1}$.  Let
$E(\varepsilon,\lambda,\eta)=\Phi(\uu)$, with $\uu$ as in \eqref{eq:u}.  Then,
$E(M)=H_p^{t_1,t_2}$, and the coordinate functions of $E$ are given by Lemma
\ref{lem:values}.  We write $f_{(a,c)}$ and $f_{(x,c,h)}$ for them.

For $\bH_p^{t_1,t_2}$, the columns of $\bA_p^{t_1,t_2}$ are the vectors
$(1,z,pc)^T$ with $z\in\Z_{p^2}^{t_1-1}$ and $c\in\Z_p^{t_2}$, by
\eqref{eq:columns}.  We denote by $\hvv_1=\one,\hvv_2,\dots,\hvv_{t_1}$ the
rows of order $p^2$ of $\bA_p^{t_1,t_2}$, and by $\hww_1,\dots,\hww_{t_2}$ its
rows of order $p$.  We index the coordinates of $\bH_p^{t_1,t_2}$ by the
triples $(z,c,h)$, where $(1,z,pc)^T$ is the corresponding column and
$h\in\Z_p$.  Let $\hat M$ be the set of pairs $(\lambda,\eta)$ with
$\lambda\in\Z_{p^2}^{t_1}$ and $\eta\in\Z_p^{t_2}$.  When $t_1=1$, the vector
$z$ is empty, and when $t_2=0$, the vectors $c$ and $\eta$ are empty.  For
$(\lambda,\eta)\in\hat M$, let
$$
\huu=\sum_{k=1}^{t_1}\lambda_k\hvv_k+\sum_{j=1}^{t_2}\eta_j\hww_j ,
$$
and let $\hat E(\lambda,\eta)=\Phi_2(\huu)$.  Then, $\hat E(\hat M)=
\bH_p^{t_1,t_2}$.  The entry of $\huu$ at the column $(1,z,pc)^T$ is
$\langle\lambda,(1,z)\rangle+p\langle\eta,c\rangle$.  Hence, by Lemma
\ref{lem:shift}, the coordinate function of $\hat E$ at $(z,c,h)$ is
\begin{equation}\label{eq:val-bar}
\hat f_{(z,c,h)}(\lambda,\eta)=\langle\eta,c\rangle
+\phi\bigl(\langle\lambda,(1,z)\rangle\bigr)_h .
\end{equation}

For $x\in\Z_{p^2}^{t_1}$, let $g_x$ be the function from $\Z_{p^2}^{t_1}$ to
$\Z_p$ which maps $\lambda$ to the digit of index $1$ of
$\langle\lambda,x\rangle$.  By \eqref{eq:gray}, $g_x(\lambda)$ is the entry of
$\phi(\langle\lambda,x\rangle)$ at the coordinate $h=0$.  For
$1\leq k\leq t_1$, we also denote by $\overline{\lambda_k}$ the function which
maps $\lambda\in\Z_{p^2}^{t_1}$ to the residue of $\lambda_k$ modulo $p$.  For
$U\subseteq\Z_{p^2}^{t_1}$, let $W_U$ be the space of functions from
$\Z_{p^2}^{t_1}$ to $\Z_p$ spanned by the functions
$\overline{\lambda_1},\dots,\overline{\lambda_{t_1}}$ and $g_x$, with $x\in U$.
Recall that $X=X_{t_1}$, and let
$$
\Omega=\{x\in\Z_{p^2}^{t_1}:\ x_1=1\}.
$$

\begin{lemma}\label{lem:W}
Let $t_1\geq1$.  Then, $\rank(H_p^{t_1,t_2})=t_2+\dim W_X$ for every
$t_2\geq1$, and $\rank(\bH_p^{t_1,t_2})=t_2+\dim W_\Omega$ for every
$t_2\geq0$.
\end{lemma}

\begin{proof}
Let $x\in\Z_{p^2}^{t_1}$.  The digit of index $0$ of $\langle\lambda,x\rangle$
is $\langle\overline{\lambda},\overline{x}\rangle$.  Hence, by \eqref{eq:gray},
we have $\phi(\langle\lambda,x\rangle)_h=g_x(\lambda)
+h\langle\overline{\lambda},\overline{x}\rangle$.  By Lemma \ref{lem:values},
it follows that every coordinate function of
$H_p^{t_1,t_2}$ is a linear combination of the functions $\varepsilon$,
$\eta_j$, $\overline{\lambda_k}$ and $g_x$, with $1\leq j\leq t_2-1$,
$1\leq k\leq t_1$ and $x\in X$.  Here, $\varepsilon$ and $\eta_j$ denote the
functions which map $(\varepsilon,\lambda,\eta)$ to these elements of $\Z_p$,
and the functions $\overline{\lambda_k}$ and $g_x$ are regarded as functions of
$\lambda$.  Conversely, we have $\varepsilon=f_{(\zero,\zero)}$, and
$$
\eta_j=f_{(\zero,e_j)}-f_{(\zero,\zero)},\quad
\overline{\lambda_k}=f_{(e_k,\zero)}-f_{(\zero,\zero)},\quad
g_x=f_{(x,\zero,0)}-f_{(\zero,\zero)}.
$$
The functions of $W_X$ only depend on $\lambda$, and they vanish when
$\lambda=\zero$.  The functions $\varepsilon$ and $\eta_j$, with
$1\leq j\leq t_2-1$, are linearly independent, and they do not depend on
$\lambda$.  Therefore, the space spanned by the coordinate functions of
$H_p^{t_1,t_2}$ is the direct sum of $W_X$ and of the space spanned by
$\varepsilon$ and the functions $\eta_j$, which has dimension $t_2$.  By Lemma
\ref{lem:evaluation}, this proves the first equality.

Similarly, by \eqref{eq:val-bar}, we have
$$
\hat f_{(z,c,h)}=\langle\eta,c\rangle+g_{(1,z)}
+h\bigl\langle\overline{\lambda},(1,\overline{z})\bigr\rangle ,
$$
where all the terms are regarded as functions of $(\lambda,\eta)$.
Conversely, $\eta_j=\hat f_{(\zero,e_j,0)}-\hat f_{(\zero,\zero,0)}$ for
$1\leq j\leq t_2$, and $g_{(1,z)}=\hat f_{(z,\zero,0)}$.  Moreover,
$\langle\overline{\lambda},(1,\overline{z})\rangle=\hat f_{(z,\zero,1)}
-\hat f_{(z,\zero,0)}$.  For $z=\zero$ and for the vectors $z$ of the
standard basis of $\Z_{p^2}^{t_1-1}$, these functions are
$\overline{\lambda_1}$ and $\overline{\lambda_1}+\overline{\lambda_k}$ with
$2\leq k\leq t_1$.  Hence, all the functions $\overline{\lambda_k}$ belong to
the space spanned by the coordinate functions of $\bH_p^{t_1,t_2}$.  As
before, this space is the direct sum of $W_\Omega$ and of the space spanned
by $\eta_1,\dots,\eta_{t_2}$, which has dimension $t_2$.  The second equality
follows.
\end{proof}

For a nonzero $\xi\in\Z_p^{t_1}$, we write $g_\xi=g_{\tilde\xi}$.  For
$U\subseteq\Z_{p^2}^{t_1}$, we put $\overline{U}=\{\overline{x}:\ x\in U\}$.

\begin{lemma}\label{lem:Wxi}
Let $U\in\{X,\Omega\}$.  Then, the space $W_U$ is spanned by the functions
$\overline{\lambda_k}$, with $1\leq k\leq t_1$, and $g_\xi$, with
$\xi\in\overline{U}$.  Moreover, $\overline{X}=\Z_p^{t_1}\setminus\{\zero\}$,
and $\overline{\Omega}$ is the set of vectors $\xi\in\Z_p^{t_1}$ with
$\xi_1=1$.
\end{lemma}

\begin{proof}
Let $x\in\Z_{p^2}^{t_1}$ and $y\in\Z_p^{t_1}$.  Since
$\langle\lambda,x+py\rangle=\langle\lambda,x\rangle+p\langle\overline{\lambda},
y\rangle$, Lemma \ref{lem:shift} gives
$$
g_{x+py}=g_x+\sum_{k=1}^{t_1}y_k\overline{\lambda_k}.
$$
By Lemma \ref{lem:X}, every vector of $X$ is of the form $\tilde\xi+py$ with
$\xi\neq\zero$, and $\tilde\xi\in X$ for every $\xi\neq\zero$.  Hence,
$\overline{X}=\Z_p^{t_1}\setminus\{\zero\}$, and the statement holds for
$U=X$.  Similarly, every vector $x\in\Omega$ is of the form $\tilde\xi+py$
with $\xi=\overline{x}$, and then $\xi_1=1$.  Conversely, $\tilde\xi\in\Omega$
for every $\xi$ with $\xi_1=1$.  Hence, $\overline{\Omega}$ is the set of
vectors $\xi$ with $\xi_1=1$, and the statement holds for $U=\Omega$.
\end{proof}

In order to compute $g_\xi$, we need the following functions.  For an
integer $w$, let $\theta(w)=(w^p-w)/p$, which is an integer by Fermat's
little theorem.  For $x\in\Z_p$, let $q(x)=\overline{\theta(\tilde x)}$,
where $\tilde x\in\{0,\dots,p-1\}$ is regarded as an integer.

Every function from $\Z_p^n$ to $\Z_p$ is given by a unique polynomial over
$\Z_p$ in which every variable has degree at most $p-1$
\cite{LidlNiederreiter}.  We call such polynomials and their monomials reduced.
For a vector $\alpha=(\alpha_1,\dots,\alpha_n)$ of nonnegative integers, we
write $y^\alpha=y_1^{\alpha_1}\cdots y_n^{\alpha_n}$,
$|\alpha|=\alpha_1+\cdots+\alpha_n$ and $\alpha!=\alpha_1!\cdots\alpha_n!$.
Thus, the monomial $y^\alpha$ is reduced if and only if $\alpha_k\leq p-1$ for
all $k$.  Since $q(0)=0$, there are unique $q_1,\dots,q_{p-1}\in\Z_p$ such that
\begin{equation}\label{eq:q}
q(x)=\sum_{d=1}^{p-1}q_dx^d
\end{equation}
for every $x\in\Z_p$.

The numerator and the denominator of a rational number are always taken in
lowest terms.  A rational number $r$ is called $p$-integral if its denominator
is not divisible by $p$.  Equivalently, $r=a/b$ for some $a,b\in\Z$ with
$p\nmid b$.  In that case, we regard $r$ as the element
$\overline{a}\,\overline{b}^{\,-1}$ of $\Z_p$.  This element does not depend on
the choice of $a$ and $b$, since $a/b=a'/b'$ implies
$\overline{a}\,\overline{b'}=\overline{a'}\,\overline{b}$.  For instance, if
$\alpha_k\leq p-1$ for all $k$, then $p$ does not divide $\alpha!$, and
$1/\alpha!$ is regarded as an element of $\Z_p$.

For $y\in\Z_p^n$, let $S(y)$ be the residue modulo $p$ of
$$
\frac1p\Bigl(\sum_{k=1}^{n}Y_k^{\,p}-\Bigl(\sum_{k=1}^{n}Y_k\Bigr)^p\Bigr),
$$
where $Y_1,\dots,Y_n$ are integers with $\overline{Y_k}=y_k$ for all $k$.  The
next lemma shows that this number is an integer, and that $S(y)$ does not
depend on the choice of the integers $Y_k$.

\begin{lemma}\label{lem:S}
Let $y\in\Z_p^n$.  Then, $S(y)$ is well defined, and
$$
S(y)=\sum_{\alpha}\frac{y^\alpha}{\alpha!},
$$
where $\alpha$ runs over the vectors of $n$ nonnegative integers with
$|\alpha|=p$ and $\alpha_k\leq p-1$ for all $k$.
\end{lemma}

\begin{proof}
Let $Y_1,\dots,Y_n$ be integers with $\overline{Y_k}=y_k$ for all $k$, and put
$Y=(Y_1,\dots,Y_n)$.  By the multinomial theorem,
$(\sum_kY_k)^p$ is the sum of the terms $(p!/\alpha!)\,Y^{\alpha}$ over the
vectors $\alpha$ with $|\alpha|=p$.  If $\alpha$ has only one nonzero entry,
then this term is $Y_k^{\,p}$ for some $k$.  Otherwise, all the entries of
$\alpha$ are at most $p-1$, and $p!/\alpha!=p\cdot(p-1)!/\alpha!$.  Here,
$(p-1)!/\alpha!$ is an integer, since $p!/\alpha!$ is an integer and $p$ does
not divide $\alpha!$.  Hence,
$$
\frac1p\Bigl(\sum_{k=1}^{n}Y_k^{\,p}-\Bigl(\sum_{k=1}^{n}Y_k\Bigr)^p\Bigr)
=-\sum_{\alpha}\frac{(p-1)!}{\alpha!}\,Y^{\alpha},
$$
where $\alpha$ runs as in the statement.  In particular, the left-hand side
is an integer, and it is given by a polynomial in $Y_1,\dots,Y_n$ with integer
coefficients.  Thus, its residue modulo $p$ does not depend on the choice of
the integers $Y_k$.  By Wilson's theorem, $(p-1)!\equiv-1\pmod p$.  Since
$\alpha!$ is not divisible by $p$, the statement follows.
\end{proof}

The monomials $y^\alpha$ in Lemma \ref{lem:S} are reduced.  The digit of
index $1$ of a sum of products is given by the following lemma.

\begin{lemma}\label{lem:digit}
Let $a,b\in\Z_p^n$.  Then, the digit of index $1$ of
$\langle\tilde a,\tilde b\rangle\in\Z_{p^2}$ is
$$
q\bigl(\langle a,b\rangle\bigr)-\sum_{k=1}^{n}\bigl(a_kq(b_k)+b_kq(a_k)\bigr)
+S(a_1b_1,\dots,a_nb_n).
$$
\end{lemma}

\begin{proof}
For $1\leq k\leq n$, let $Y_k=\tilde a_k\tilde b_k$, as an integer.  Let
$w=\sum_kY_k$, and write $w=w_0+pw_1$ with integers $w_0\in\{0,\dots,p-1\}$ and
$w_1$.  The element $\langle\tilde a,\tilde b\rangle$ of $\Z_{p^2}$ is the
residue of $w$ modulo $p^2$, and its digit of index $1$ is $\overline{w_1}$.
Since $w^p\equiv w_0^p\pmod{p^2}$, we have
$\theta(w)\equiv\theta(w_0)-w_1\pmod p$.  Moreover,
$\overline{w_0}=\overline{w}=\langle a,b\rangle$, and hence
$w_0=\widetilde{\langle a,b\rangle}$ and
$\overline{\theta(w_0)}=q(\langle a,b\rangle)$.  Therefore, the digit is
$q(\langle a,b\rangle)-\overline{\theta(w)}$.  By the definition of $\theta$,
we have
$$
\theta(w)=\sum_{k=1}^{n}\theta(Y_k)
-\frac1p\Bigl(\sum_{k=1}^{n}Y_k^p-w^p\Bigr).
$$
The residue of $(\sum_kY_k^p-w^p)/p$ is $S(a_1b_1,\dots,a_nb_n)$, since
$\overline{Y_k}=a_kb_k$ for all $k$.  Finally,
$\theta(Y_k)=\theta(\tilde a_k)\tilde b_k^{\,p}+\tilde a_k\theta(\tilde b_k)$,
and $\tilde b_k^{\,p}\equiv\tilde b_k\pmod p$.  Since
$\overline{\theta(\tilde a_k)}=q(a_k)$ and
$\overline{\theta(\tilde b_k)}=q(b_k)$, the statement follows.
\end{proof}

The following set will describe the rank of $H_p^{t_1,t_2}$.  Let
\begin{equation}\label{eq:Dp}
D_p=\{0,1,p\}\cup\{d:\ 2\leq d\leq p-1 \text{ and } q_d\neq0\}.
\end{equation}

Every $\lambda\in\Z_{p^2}^{t_1}$ can be written in a unique way as
$\lambda=\widetilde{\lambda'}+p\lambda''$ with
$\lambda',\lambda''\in\Z_p^{t_1}$.  Thus, $\lambda'_k$ and $\lambda''_k$ are
the digits of index $0$ and $1$ of $\lambda_k$, and
$\overline{\lambda_k}=\lambda'_k$.  This gives a bijection between
$\Z_{p^2}^{t_1}$ and $\Z_p^{t_1}\times\Z_p^{t_1}$.  Hence, we regard the
functions of $W_X$ and $W_\Omega$ as functions of $(\lambda',\lambda'')$, that
is, as reduced polynomials in the $2t_1$ variables
$\lambda'_1,\dots,\lambda'_{t_1},\lambda''_1,\dots,\lambda''_{t_1}$.  Let $L$
be the space spanned by $\lambda'_1,\dots,\lambda'_{t_1}$.  Let $\alpha$ be a
vector of $t_1$ nonnegative integers with $2\leq|\alpha|\leq p$ and
$\alpha_k\leq p-1$ for all $k$, and let $\xi\in\Z_p^{t_1}$.  We put
$$
\tau_\alpha(\xi)=\begin{cases}
q_d\,\dfrac{d!}{\alpha!}\,\xi^\alpha & \text{if } |\alpha|=d\leq p-1
\text{ and } \alpha\neq de_k \text{ for all } k,\\[2mm]
q_d\,(\xi_k^d-\xi_k) & \text{if } \alpha=de_k \text{ with } d\leq p-1,
\\[2mm]
\xi^\alpha/\alpha! & \text{if } |\alpha|=p.
\end{cases}
$$

\begin{lemma}\label{lem:gxi}
Let $\xi\in\Z_p^{t_1}$ be nonzero.  Then,
$$
g_\xi\equiv\sum_{\alpha}\tau_\alpha(\xi)\,{\lambda'}^{\alpha}
+\sum_{k=1}^{t_1}\xi_k\lambda''_k\pmod{L},
$$
where $\alpha$ runs over the vectors of $t_1$ nonnegative integers with
$2\leq|\alpha|\leq p$ and $\alpha_k\leq p-1$ for all $k$.
\end{lemma}

\begin{proof}
Since $\langle\lambda,\tilde\xi\rangle
=\langle\widetilde{\lambda'},\tilde\xi\rangle+p\langle\lambda'',\xi\rangle$,
Lemma \ref{lem:shift} shows that $g_\xi(\lambda)$ is the sum of the digit of
index $1$ of $\langle\widetilde{\lambda'},\tilde\xi\rangle\in\Z_{p^2}$ and of
$\langle\lambda'',\xi\rangle$.  We apply Lemma \ref{lem:digit} with $a=\xi$ and
$b=\lambda'$, and we use \eqref{eq:q} and Lemma \ref{lem:S}.  We obtain
$$
g_\xi=\sum_{d=2}^{p-1}q_d\Bigl(\langle\lambda',\xi\rangle^d
-\sum_{k=1}^{t_1}\xi_k{\lambda'_k}^{d}\Bigr)
+\sum_{|\alpha|=p}\frac{\xi^\alpha}{\alpha!}\,{\lambda'}^{\alpha}
+\langle\lambda'',\xi\rangle-\sum_{k=1}^{t_1}q(\xi_k)\lambda'_k ,
$$
where the second sum runs over the vectors $\alpha$ with $|\alpha|=p$ and
$\alpha_k\leq p-1$ for all $k$.  The term with $d=1$ does not appear, since
$\langle\lambda',\xi\rangle-\sum_k\xi_k\lambda'_k=0$.  The last sum belongs to
$L$.  Let $2\leq d\leq p-1$.  By the multinomial theorem,
$$
\langle\lambda',\xi\rangle^d=\sum_{|\alpha|=d}\frac{d!}{\alpha!}\,
\xi^\alpha{\lambda'}^{\alpha},
$$
and the monomials ${\lambda'}^{\alpha}$ with $|\alpha|=d$ are reduced.
Moreover, the coefficient of ${\lambda'_k}^{d}$ in
$\langle\lambda',\xi\rangle^d-\sum_j\xi_j{\lambda'_j}^{d}$ is $\xi_k^d-\xi_k$.
Hence, the statement follows from the definition of $\tau_\alpha$.
\end{proof}

\begin{lemma}\label{lem:F}
Let $\cF$ be the space of functions from $\Z_p^{t_1}$ to $\Z_p$ spanned by the
functions $\tau_\alpha$, with $\alpha$ as in Lemma \ref{lem:gxi}, and by the
functions $\xi\mapsto\xi_k$, with $1\leq k\leq t_1$.  Then, the reduced
monomials $\xi^\alpha$ with $|\alpha|\in D_p\setminus\{0\}$ form a basis of
$\cF$.
\end{lemma}

\begin{proof}
Let $2\leq d\leq p-1$.  If $q_d=0$, then $\tau_\alpha=0$ for every $\alpha$
with $|\alpha|=d$.  Assume that $q_d\neq0$.  For $|\alpha|=d$ and
$\alpha\neq de_k$, the function $\tau_\alpha$ is a nonzero multiple of
$\xi^\alpha$, since $d!/\alpha!$ is not divisible by $p$.  For $\alpha=de_k$,
we have $\tau_\alpha=q_d(\xi_k^d-\xi_k)$.  Hence, the functions $\tau_\alpha$
with $|\alpha|=d$, together with $\xi_1,\dots,\xi_{t_1}$, span the same space
as the monomials of degree $d$ and $\xi_1,\dots,\xi_{t_1}$.  These monomials
are reduced, since $d\leq p-1$.  Finally, the functions $\tau_\alpha$ with
$|\alpha|=p$ are nonzero multiples of the reduced monomials of degree $p$.
By \eqref{eq:Dp}, it follows that $\cF$ is spanned by the reduced monomials
$\xi^\alpha$ with $|\alpha|\in D_p\setminus\{0\}$.  These monomials are
linearly independent, and the statement follows.
\end{proof}

\begin{theorem}\label{thm:rank}
Let $p$ be a prime, and let $t_1\geq1$ and $t_2\geq0$.  Then,
$$
\rank\bigl(\bH_p^{t_1,t_2}\bigr)=t_2+1+\binom{t_1+p-1}{p}.
$$
If $t_2\geq1$, then
$$
\rank\bigl(H_p^{t_1,t_2}\bigr)=t_2-1+\sum_{d\in D_p}\binom{t_1+d-1}{d},
$$
where $D_p$ is the set given by \eqref{eq:Dp}.
\end{theorem}

\begin{proof}
By Lemma \ref{lem:W}, we have to show that $\dim W_\Omega=1+\binom{t_1+p-1}{p}$
and that $\dim W_X$ is the sum of $\binom{t_1+d-1}{d}$ over
$d\in D_p\setminus\{0\}$.  Let $U\in\{X,\Omega\}$.  By Lemma \ref{lem:Wxi}, the
space $W_U$ is spanned by $L$ and by the functions $g_\xi$ with
$\xi\in\overline{U}$.  By Lemma \ref{lem:gxi}, modulo $L$, every such function
$g_\xi$ is a linear combination of the functions ${\lambda'}^{\alpha}$ and
$\lambda''_k$, with the coefficients $\tau_\alpha(\xi)$ and $\xi_k$.  The
functions ${\lambda'}^{\alpha}$, $\lambda'_k$ and $\lambda''_k$ are pairwise
different reduced monomials in the $2t_1$ variables, and hence they are
linearly independent.  Therefore, $\dim W_U=t_1+r_U$, where $r_U$ is the rank
of the matrix whose rows are indexed by $\xi\in\overline{U}$, and whose row
$\xi$ is formed by the values $\tau_\alpha(\xi)$ and $\xi_1,\dots,\xi_{t_1}$.
As in the proof of Lemma \ref{lem:evaluation}, the rank $r_U$ is also the
dimension of the space spanned by the columns of this matrix, regarded as
functions on $\overline{U}$.  These columns are the restrictions to
$\overline{U}$ of the functions which span the space $\cF$ of Lemma
\ref{lem:F}.  Hence, $r_U$ is the dimension of the space formed by the
restrictions to $\overline{U}$ of the functions of $\cF$.

First, let $U=X$.  By Lemma \ref{lem:F}, the functions of $\cF$ are given by
polynomials without constant term.  By Lemma \ref{lem:Wxi},
$\overline{X}=\Z_p^{t_1}\setminus\{\zero\}$.  Hence, a function of $\cF$ which
vanishes on $\overline{X}$ also vanishes at $\zero$, and it is zero.
Therefore, $r_X=\dim\cF$.  The number of reduced monomials of degree $d$ in
$t_1$ variables is $\binom{t_1+d-1}{d}$ if $d\leq p-1$, and
$\binom{t_1+p-1}{p}-t_1$ if $d=p$.  By Lemma \ref{lem:F}, we obtain
$$
\dim W_X=t_1+\sum_{d\in D_p\setminus\{0,p\}}\binom{t_1+d-1}{d}
+\binom{t_1+p-1}{p}-t_1=\sum_{d\in D_p\setminus\{0\}}\binom{t_1+d-1}{d}.
$$

Now, let $U=\Omega$.  By Lemma \ref{lem:Wxi}, the vectors of
$\overline{\Omega}$ are the vectors $\xi=(1,\xi')$ with
$\xi'=(\xi_2,\dots,\xi_{t_1})\in\Z_p^{t_1-1}$.  The restriction of a monomial
$\xi^\alpha$ to $\overline{\Omega}$ is a monomial in $\xi'$ of degree at most
$|\alpha|$, and it is reduced when $\xi^\alpha$ is reduced.  Conversely, let
${\xi'}^{\delta}$ be a reduced monomial with $1\leq|\delta|\leq p$.  Then,
${\xi'}^{\delta}$ is the restriction of $\xi_1^{p-|\delta|}{\xi'}^{\delta}$,
which is a reduced monomial of degree $p$.  Moreover, the constant function $1$
is the restriction of $\xi_1$.  Since $1,p\in D_p$, Lemma \ref{lem:F} shows
that the restrictions of the functions of $\cF$ to $\overline{\Omega}$ form the
space spanned by the reduced monomials in $\xi'$ of degree at most $p$.  These
monomials are linearly independent.  Their number is
$\binom{t_1+p-1}{p}-(t_1-1)$, since the monomials of degree at most $p$ in
$\xi'$ which are not reduced are $\xi_2^p,\dots,\xi_{t_1}^p$.  Therefore,
$\dim W_\Omega=t_1+\binom{t_1+p-1}{p}-(t_1-1)$, and the proof is complete.
\end{proof}

For $p=2$, there is no $d$ with $2\leq d\leq p-1$, and hence
$D_2=\{0,1,2\}$.  In this case, Theorem \ref{thm:rank} gives
$$
\rank\bigl(H_2^{t_1,t_2}\bigr)=2t_1+t_2+\binom{t_1}{2}\quad\text{and}\quad
\rank\bigl(\bH_2^{t_1,t_2}\bigr)=2t_1+t_2+\binom{t_1-1}{2},
$$
which are the formulas of \cite{PRV06}.  For odd $p$, the set $D_p$ is
described by the Bernoulli numbers $B_m$, which are the rational numbers
given by
$$
\frac{x}{e^x-1}=\sum_{m\geq0}B_m\frac{x^m}{m!}.
$$
We have $B_0=1$, $B_1=-1/2$, and $B_m=0$ for every odd $m\geq3$.  For two
$p$-integral rational numbers $r$ and $r'$, we write $r\equiv r'\pmod{p^k}$ if
$(r-r')/p^k$ is $p$-integral.  For $n\geq0$, let
$\sigma_n=\sum_{x=1}^{p-1}x^n\in\Z$.  We will use the following classical
results.

\begin{lemma}[\cite{IrelandRosen,Washington}]\label{lem:bernoulli}
Let $p$ be an odd prime.  For every even $m\geq2$, the number $pB_m$ is
$p$-integral, and $B_m$ is $p$-integral if and only if $p-1$ does not divide
$m$.  For every $n\geq1$, we have
$$
\sigma_n=\sum_{i=0}^{n}\binom{n}{i}\frac{B_i\,p^{n+1-i}}{n+1-i}.
$$
Moreover, if $m,m'\geq2$ are even, $m\equiv m'\pmod{p-1}$ and $p-1$ does not
divide $m$, then $B_m/m$ and $B_{m'}/m'$ are $p$-integral, and
$B_{m'}/m'\equiv B_m/m\pmod p$.
\end{lemma}

The first statement of Lemma \ref{lem:bernoulli} is the theorem of von Staudt
and Clausen, and the last one is Kummer's congruence.

\begin{proposition}\label{prop:coeff}
Let $p$ be an odd prime, and let $q_1,\dots,q_{p-1}$ be as in \eqref{eq:q}.
Then, $q_{p-1}=(p-1)/2$, and $q_d=0$ for every even $d$ with
$2\leq d\leq p-3$.  Moreover, $q_d=-B_{p-d}/d$ for every odd $d$ with
$3\leq d\leq p-2$.
\end{proposition}

\begin{proof}
Since the group $\Z_p^*$ is cyclic of order $p-1$, the sum $\sum_{x\in\Z_p}x^m$
is equal to $-1$ if $m$ is a positive multiple of $p-1$, and to $0$ otherwise,
where $0^0=1$.  Let $1\leq d\leq p-1$.  For $1\leq j\leq p-1$, the exponent
$j+p-1-d$ is a positive multiple of $p-1$ if and only if $j=d$.  Hence, if we
multiply \eqref{eq:q} by $x^{p-1-d}$ and sum over $x\in\Z_p$, we obtain
$q_d=-\sum_{x\in\Z_p}q(x)x^{p-1-d}$.  We have
$q(\overline{x})=\overline{\theta(x)}$ for every integer $x$ with
$0\leq x\leq p-1$, and $\theta(0)=0$.  Moreover,
$\sum_{x=1}^{p-1}\theta(x)x^{p-1-d}=(\sigma_{2p-1-d}-\sigma_{p-d})/p$, and this
number is an integer.  Hence, in $\Z_p$,
\begin{equation}\label{eq:cd}
q_d=-\frac{\sigma_{2p-1-d}-\sigma_{p-d}}{p}.
\end{equation}

Let $m$ be odd.  For $1\leq x\leq p-1$, we have
$x^m+(p-x)^m\equiv mpx^{m-1}\pmod{p^2}$.  Hence,
$\sigma_m\equiv mp\sum_{x=1}^{(p-1)/2}x^{m-1}\pmod{p^2}$.  If $m-1$ is not a
multiple of $p-1$, then $2\sum_{x=1}^{(p-1)/2}x^{m-1}\equiv\sigma_{m-1}\equiv0
\pmod p$, and hence $\sigma_m\equiv0\pmod{p^2}$.  If $d$ is even and
$2\leq d\leq p-3$, then $p-d$ and $2p-1-d$ are odd, and neither $p-1-d$ nor
$2p-2-d$ is a multiple of $p-1$.  Hence, $q_d=0$, by \eqref{eq:cd}.  If
$d=p-1$, then $\sigma_1=p(p-1)/2$ and
$\sigma_p\equiv p^2\sum_{x=1}^{(p-1)/2}x^{p-1}\equiv0\pmod{p^2}$.  Hence,
$q_{p-1}=(p-1)/2$.

Finally, let $d$ be odd with $3\leq d\leq p-2$.  Then, $p\geq5$, and
$m=p-d$ is even with $2\leq m\leq p-3$.  Let $m'=m+p-1\leq2p-4$.  Neither $m$
nor $m'$ is a multiple of $p-1$.  We first show that
$\sigma_n\equiv pB_n\pmod{p^2}$ for $n\in\{m,m'\}$.  By Lemma
\ref{lem:bernoulli}, $B_n$ is $p$-integral, and the term with $i=n$ in the
formula for $\sigma_n$ is $pB_n$.  Let $0\leq i<n$.  If $i\neq p-1$ and
$n+1-i\neq p$, then $B_i$ is $p$-integral, and $p$ does not divide $n+1-i$,
since $2\leq n+1-i\leq 2p-3$.  Hence, the term with index $i$ is $p^2$ times a
$p$-integral number.  If $n+1-i=p$, then $n=m'$ and $i=m$, and this term is
$\binom{m'}{m}B_m\,p^{p-1}$.  If $i=p-1$, then $n=m'$, and this term is
$\binom{m'}{p-1}(pB_{p-1})\,p^{m}/(m+1)$, where $m\geq2$ and $p$ does not
divide $m+1$.  In both cases, this term is again $p^2$ times a $p$-integral
number.  Therefore, $\sigma_m\equiv pB_m$ and
$\sigma_{m'}\equiv pB_{m'}\pmod{p^2}$.  By Lemma \ref{lem:bernoulli}, we also
have $B_{m'}/m'\equiv B_m/m\pmod p$.  Since $m'\equiv m-1\pmod p$, it follows
that $B_{m'}\equiv(m-1)B_m/m\pmod p$.  Hence, \eqref{eq:cd} gives
$q_d=B_m-B_{m'}=B_m/m=-B_{p-d}/d$ in $\Z_p$, since $m\equiv-d\pmod p$.
\end{proof}

Since $B_m=0$ for every odd $m\geq3$, and since $-B_1/(p-1)=(p-1)/2$ in $\Z_p$,
Proposition \ref{prop:coeff} can also be written as $q_d=-B_{p-d}/d$ for every
$d$ with $2\leq d\leq p-1$.

An odd prime $p$ is called regular if it does not divide the numerator of $B_m$
for any even $m$ with $2\leq m\leq p-3$, and irregular otherwise
\cite{Washington}.  The number of even $m$ with $2\leq m\leq p-3$ such that $p$
divides the numerator of $B_m$ is called the index of irregularity of $p$, and
we denote it by $i_p$.  Since $d$ is invertible in $\Z_p$ for
$1\leq d\leq p-1$, Proposition \ref{prop:coeff} and \eqref{eq:Dp} give the
following result.

\begin{corollary}\label{cor:Dp}
Let $p$ be an odd prime.  Then, $D_p$ is obtained from $\{0,1,\dots,p\}$ by
removing the even numbers $d$ with $2\leq d\leq p-3$, and the numbers $p-m$
such that $m$ is even, $2\leq m\leq p-3$, and $p$ divides the numerator of
$B_m$.  In particular, if $p$ is regular, then $D_p$ is formed by $0$, $p-1$
and the odd numbers $d$ with $1\leq d\leq p$.
\end{corollary}

For regular primes, Theorem \ref{thm:rank} and Corollary \ref{cor:Dp} give a
closed formula for the rank of $H_p^{t_1,t_2}$.

\begin{corollary}\label{cor:regular}
Let $p$ be a regular odd prime, and let $t_1,t_2\geq1$.  Then,
$$
\rank\bigl(H_p^{t_1,t_2}\bigr)=t_2-1+\binom{t_1+p}{p}
-\sum_{j=1}^{(p-3)/2}\binom{t_1+2j-1}{2j}.
$$
\end{corollary}

\begin{proof}
We have $\sum_{d=0}^{p}\binom{t_1+d-1}{d}=\binom{t_1+p}{p}$.  Hence, the
statement follows from Theorem \ref{thm:rank} and Corollary \ref{cor:Dp}.
\end{proof}

For $t_1=1$ and every odd prime $p$, the rank of $H_p^{1,t_2}$ only depends on
the number of elements of $D_p$, which gives the following closed formula.

\begin{corollary}\label{cor:t1}
Let $p$ be an odd prime, and let $t_2\geq1$.  Then,
$$
\rank\bigl(H_p^{1,t_2}\bigr)=t_2+\frac{p+3}{2}-i_p .
$$
\end{corollary}

\begin{proof}
For $t_1=1$, every binomial coefficient in Theorem \ref{thm:rank} is equal to
$1$.  Hence, $\rank(H_p^{1,t_2})=t_2-1+|D_p|$.  The numbers $p-m$ of
Corollary \ref{cor:Dp} are odd and pairwise different.  Hence,
$|D_p|=(p+1)-(p-3)/2-i_p=(p+5)/2-i_p$, and the statement follows.
\end{proof}

The following consequence of Theorem \ref{thm:rank} is the key ingredient
of the proof of the main result.

\begin{corollary}\label{cor:rank}
Let $p$ be an odd prime, and let $t_1,t_2\geq1$.  Then,
$$
\rank\bigl(\bH_p^{t_1,t_2}\bigr)<\rank\bigl(H_p^{t_1,t_2}\bigr).
$$
\end{corollary}

\begin{proof}
By Theorem \ref{thm:rank}, $\rank(H_p^{t_1,t_2})-\rank(\bH_p^{t_1,t_2})$ is
the sum of $\binom{t_1+d-1}{d}$ over $d\in D_p\setminus\{p\}$, minus $2$.  By
\eqref{eq:Dp} and Proposition \ref{prop:coeff}, the set $D_p$ contains $0$,
$1$ and $p-1$.  Hence, this difference is at least
$1+t_1+\binom{t_1+p-2}{p-1}-2$, which is positive.
\end{proof}

\begin{example}\label{ex:rank}
We continue Example \ref{ex:blocks}, with $p=3$ and $(t_1,t_2)=(2,1)$.  We have
$\theta(0)=\theta(1)=0$ and $\theta(2)=(8-2)/3=2$.  Thus, $q(0)=q(1)=0$ and
$q(2)=2$.  Hence, $q(x)=2x+x^2$, and $D_3=\{0,1,2,3\}$.  By Lemma \ref{lem:F},
the space $\cF$ has the basis formed by the seven reduced monomials $\xi_1$,
$\xi_2$, $\xi_1^2$, $\xi_1\xi_2$, $\xi_2^2$, $\xi_1^2\xi_2$ and $\xi_1\xi_2^2$.
Hence, by the proof of Theorem \ref{thm:rank}, we have $\dim W_X=2+7=9$.  On
the set $\overline{\Omega}$ of the vectors $(1,\xi_2)$, the restrictions of
these monomials span the space with basis $1$, $\xi_2$ and $\xi_2^2$.  Hence,
$\dim W_\Omega=2+3=5$.  By Lemma \ref{lem:W}, the ranks of $H_3^{2,1}$ and
$\bH_3^{2,1}$ are $1+9=10$ and $1+5=6$, respectively.  This agrees with Theorem
\ref{thm:rank} and with the values of \cite{ZpZp2Classification}.

For $p=5$, we have $q(x)=3x^3+2x^4$ and $D_5=\{0,1,3,4,5\}$.  This agrees with
Proposition \ref{prop:coeff}, since $q_2=0$, $q_4=2=(5-1)/2$ and $q_3=3=-B_2/3$
in $\Z_5$, where $B_2=1/6$.  The smallest irregular prime is $37$.  It divides
the numerator of $B_{32}$, and of no other $B_m$ with $m$ even and
$2\leq m\leq34$ \cite{Washington}.  Hence, $i_{37}=1$, the set $D_{37}$ does
not contain $37-32=5$, and $\rank(H_{37}^{1,t_2})=t_2+19$ by Corollary
\ref{cor:t1}.
\end{example}

Table \ref{tab:invariants} collects the invariants of $H_p^{t_1,t_2}$ and
$\bH_p^{t_1,t_2}$ considered so far, for odd $p$ and $t_1\geq2$.  They are
given by \eqref{eq:length} and \eqref{eq:t}, by Theorems \ref{thm:kerZps},
\ref{thm:kerZpZp2}, \ref{thm:derivedZps}, \ref{thm:blocks}, \ref{thm:derived}
and \ref{thm:rank}, and by Corollary \ref{cor:profile} and the discussion
after Theorem \ref{thm:derivedZps}.  All these invariants coincide except the
rank, and the ranks are different by Corollary \ref{cor:rank}.

\begin{table}[!ht]
\small
\renewcommand{\arraystretch}{1.3}
\centering
\caption{The invariants of $H_p^{t_1,t_2}$ and $\bH_p^{t_1,t_2}$, for odd $p$
and $t_1\geq2$.}
\label{tab:invariants}
\vspace{1.5mm}
\begin{tabular}{|l|c|c|}
\hline
& $H_p^{t_1,t_2}$ & $\bH_p^{t_1,t_2}$\\
\hline\hline
length & $p^{2t_1+t_2-1}$ & $p^{2t_1+t_2-1}$\\ \hline
dimension of the kernel & $t_1+t_2$ & $t_1+t_2$\\ \hline
number of kernel blocks & $p^{t_1+t_2-1}$ & $p^{t_1+t_2-1}$\\ \hline
size of the kernel blocks & $p^{t_1}$ & $p^{t_1}$\\ \hline
derived code & $\bH_p^{t_1+1}$ & $\bH_p^{t_1+1}$\\ \hline
kernel-block rank profile & $t_1+1$, repeated $p^{t_1+t_2-1}$ times
& $t_1+1$, repeated $p^{t_1+t_2-1}$ times\\ \hline
rank & $t_2-1+\sum_{d\in D_p}\binom{t_1+d-1}{d}$
& $t_2+1+\binom{t_1+p-1}{p}$\\ \hline
\end{tabular}
\end{table}

Table \ref{tab:ranks} gives the ranks of $H_p^{t_1,t_2}$ and
$\bH_p^{t_1,t_2}$ for $p\in\{3,5\}$, $t_1\geq2$ and $t\leq8$, by Theorem
\ref{thm:rank}.  The difference between the two ranks grows quickly with
$t_1$.

\begin{table}[!ht]
\small
\renewcommand{\arraystretch}{1.15}
\centering
\caption{The ranks of $H_p^{t_1,t_2}$ and $\bH_p^{t_1,t_2}$ for
$p\in\{3,5\}$.}
\label{tab:ranks}
\vspace{1.5mm}
\begin{tabular}{|c|c||c|c||c|c|}
\hline
\multirow{2}{*}{$t$} & \multirow{2}{*}{$(t_1,t_2)$}
& \multicolumn{2}{c||}{$p=3$} & \multicolumn{2}{c|}{$p=5$}\\
\cline{3-6}
& & $H_3^{t_1,t_2}$ & $\bH_3^{t_1,t_2}$ & $H_5^{t_1,t_2}$ & $\bH_5^{t_1,t_2}$\\
\hline\hline
$4$ & $(2,1)$ & $10$ & $6$ & $18$ & $8$\\ \hline
$5$ & $(2,2)$ & $11$ & $7$ & $19$ & $9$\\ \hline
\multirow{2}{*}{$6$} & $(2,3)$ & $12$ & $8$ & $20$ & $10$\\ \cline{2-6}
& $(3,1)$ & $20$ & $12$ & $50$ & $23$\\ \hline
\multirow{2}{*}{$7$} & $(2,4)$ & $13$ & $9$ & $21$ & $11$\\ \cline{2-6}
& $(3,2)$ & $21$ & $13$ & $51$ & $24$\\ \hline
\multirow{3}{*}{$8$} & $(2,5)$ & $14$ & $10$ & $22$ & $12$\\ \cline{2-6}
& $(3,3)$ & $22$ & $14$ & $52$ & $25$\\ \cline{2-6}
& $(4,1)$ & $35$ & $22$ & $116$ & $58$\\ \hline
\end{tabular}
\end{table}

\section{The main result for odd primes}\label{sec:main}

In this section, $p$ is odd.  We prove that the codes $H_p^{t_1,t_2}$ are not
equivalent to any $\Z_{p^s}$-linear GH code, and we count the equivalence
classes of the union of both families.  The derived code reduces the
comparison to the pairs of codes $H_p^{t_1,t_2}$ and $\bH_p^{t_1,t_2}$, and the
rank completes the proof.

\begin{theorem}\label{thm:main}
Let $p$ be an odd prime, and let $t_1,t_2\geq1$.  Then, the
$\Z_p\Z_{p^2}$-linear GH code $H_p^{t_1,t_2}$ is not equivalent to any
$\Z_{p^s}$-linear GH code of the same length, for any $s\geq1$.
\end{theorem}

\begin{proof}
Let $H=H_p^{t_1,t_2}$, and let $p^t$ be its length, where $t=2t_1+t_2-1$.
By Theorem \ref{thm:linZpZp2}, the code $H$ is nonlinear.  Assume that $H$ is
equivalent to a $\Z_{p^s}$-linear GH code of length $p^t$.  By Proposition
\ref{prop:candidates}, the code $H$ is equivalent to a nonlinear reduced code
$\bH_p^{u_1,\dots,u_r}$ of length $p^t$ whose derived code is linear and has
length $p^{t_1}$.  By Theorem \ref{thm:derivedZps} and Lemma
\ref{lem:invariance}, it follows that the code
$\bH_p^{u_1,\dots,u_{r-2},u_{r-1}+1}$ is linear, and that its length is
$p^{t_1}$.  By the discussion after Theorem \ref{thm:derivedZps}, we obtain
$r=2$.  Then, the code $\bH_p^{u_1+1}$ has length $p^{u_1}$, and hence
$u_1=t_1$.  Since $\bH_p^{u_1,u_2}$ has length $p^t$, \eqref{eq:length} gives
$2u_1+u_2=t+1=2t_1+t_2$, and hence $u_2=t_2$.

Therefore, $t_1=u_1\geq2$, and $H_p^{t_1,t_2}$ is equivalent to
$\bH_p^{t_1,t_2}$.  By Lemma \ref{lem:invariance}, these two codes have the
same rank.  This contradicts Corollary \ref{cor:rank}.
\end{proof}

We now count the equivalence classes.

\begin{corollary}\label{cor:count}
Let $p$ be an odd prime, and let $t\geq2$.  Then, the codes $H_p^{t_1,t_2}$
of length $p^t$, together with the $\Z_{p^s}$-linear GH codes of length $p^t$
with $s\geq1$, form exactly
$$
\cP(t+1)-\cP(t)+\left\lfloor\frac{t}{2}\right\rfloor
$$
equivalence classes.  A complete system of representatives is formed by the
codes $H_p^{t_1,t+1-2t_1}$ with $1\leq t_1\leq\lfloor t/2\rfloor$, the
linear code $\bH_p^{t+1}$, and the reduced codes $\bH_p^{u_1,\dots,u_r}$ of
length $p^t$.
\end{corollary}

\begin{proof}
By Corollary \ref{cor:classZpZp2}, the codes $H_p^{t_1,t_2}$ of length $p^t$
form $\lfloor t/2\rfloor$ equivalence classes, represented by the codes
$H_p^{t_1,t+1-2t_1}$ with $1\leq t_1\leq\lfloor t/2\rfloor$.  If $t\geq3$,
then, by Theorem \ref{thm:classZps}, the $\Z_{p^s}$-linear GH codes of length
$p^t$ form $\cP(t+1)-\cP(t)$ equivalence classes, represented by $\bH_p^{t+1}$
and the reduced codes of length $p^t$.  If $t=2$, then there is no reduced code
of length $p^2$, since every reduced code $\bH_p^{u_1,\dots,u_r}$ of length
$p^t$ satisfies $t+1\geq ru_1\geq4$ by \eqref{eq:length}.  Hence, by
Proposition \ref{prop:reduction}, every $\Z_{p^s}$-linear GH code of length
$p^2$ with $s\geq2$ is linear and permutation equivalent to $\bH_p^{3}$.  Thus,
the $\Z_{p^s}$-linear GH codes of length $p^2$ form one equivalence class, and
$\cP(3)-\cP(2)=3-2=1$.  In both cases, Theorem \ref{thm:main} shows that no
code $H_p^{t_1,t_2}$ is equivalent to a $\Z_{p^s}$-linear GH code.  Therefore,
the two families of classes are disjoint, and the statement follows.
\end{proof}

For $3\leq t\leq14$, these numbers are given in Table \ref{tab:count} of
Section \ref{sec:binary}, together with the corresponding numbers for $p=2$.

\begin{example}\label{ex:t6}
Let $p$ be odd, and let $t=6$.  The codes $H_p^{t_1,t_2}$ of length $p^6$ are
$H_p^{1,5}$, $H_p^{2,3}$ and $H_p^{3,1}$, and their kernels have dimension
$6$, $5$ and $4$, respectively.  By Theorem \ref{thm:classZps}, the
$\Z_{p^s}$-linear GH codes of length $p^6$ are represented by the linear code
$\bH_p^{7}$ and by the reduced codes $\bH_p^{2,3}$, $\bH_p^{3,1}$ and
$\bH_p^{2,0,1}$.  By Theorem \ref{thm:kerZps}, the kernels of these reduced
codes have dimension $5$, $4$ and $3$, respectively.  The linear code
$\bH_p^{7}$ has kernel dimension $7$.  Hence, by Lemma \ref{lem:invariance},
the only possible equivalences are between $H_p^{2,3}$ and $\bH_p^{2,3}$, and
between $H_p^{3,1}$ and $\bH_p^{3,1}$.  By Theorem \ref{thm:derived}, the
derived codes of $H_p^{2,3}$ and $H_p^{3,1}$ are the linear codes $\bH_p^{3}$
and $\bH_p^{4}$, respectively.  By Theorem \ref{thm:derivedZps}, these are
also the derived codes of $\bH_p^{2,3}$ and $\bH_p^{3,1}$, respectively.
Hence, the kernel blocks do not separate these pairs, and we use Corollary
\ref{cor:rank}.  For instance, for $p=3$, the ranks are $12$ and $8$ for the
first pair, and $20$ and $12$ for the second one, by Table \ref{tab:ranks}.
Therefore, there are $3+4=7$ equivalence classes, in accordance with Table
\ref{tab:count}.  For an arbitrary odd prime $p$, the results of
\cite{ZpZp2Classification} cover the codes $H_p^{1,5}$ and $H_p^{2,3}$, but
not the code $H_p^{3,1}$.
\end{example}

\section{The binary case}\label{sec:binary}

In this section, $p=2$.  By Theorem \ref{thm:linZpZp2}, a code
$H_2^{t_1,t_2}$ is nonlinear if and only if $t_1\geq2$.  Unlike the case of
odd primes, many of these codes are equivalent to $\Z_4$-linear Hadamard
codes.  The following result is a consequence of \cite{KV2015} and of the
classification of the $\Z_2\Z_4$-linear Hadamard codes with $\alpha_1\neq0$
given in \cite{PRV06}.  We include its short proof, since these codes are
described in a different way in \cite{KV2015}.

\begin{theorem}[\cite{KV2015,PRV06}]\label{thm:KV}
Let $t_1\geq1$ and $t_2\geq2$.  Then, $H_2^{t_1,t_2}$ is equivalent to the
$\Z_4$-linear Hadamard code $\bH_2^{t_1+1,t_2-2}$.
\end{theorem}

\begin{proof}
By \eqref{eq:length} and \eqref{eq:t}, the codes $\bH_2^{t_1+1,t_2-2}$ and
$H_2^{t_1,t_2}$ have the same length $2^t$, with $t=2t_1+t_2-1\geq3$.  If
$t_1\geq2$, then both codes are nonlinear, and their kernels have dimension
$t_1+t_2$, by Theorems \ref{thm:linZps}, \ref{thm:kerZps}, \ref{thm:linZpZp2}
and \ref{thm:kerZpZp2}.  If $t_1=1$, then both codes are linear, by Theorems
\ref{thm:linZps} and \ref{thm:linZpZp2}, and hence their kernels have dimension
$t+1$.  By \cite{KV2015}, the code $\bH_2^{t_1+1,t_2-2}$ is equivalent to a
$\Z_2\Z_4$-linear Hadamard code $C$ with $\alpha_1\neq0$.  There are exactly
$\lfloor t/2\rfloor$ such codes of length $2^t$, up to equivalence
\cite{PRV06}.  Hence, by Corollary \ref{cor:classZpZp2}, the code $C$ is
equivalent to a code $H_2^{t'_1,t'_2}$ of length $2^t$.  By Lemma
\ref{lem:invariance}, the kernel of $H_2^{t'_1,t'_2}$ has the same dimension as
the kernel of $\bH_2^{t_1+1,t_2-2}$, and hence as the kernel of
$H_2^{t_1,t_2}$.  By the discussion before Corollary \ref{cor:classZpZp2}, the
kernels of the codes $H_2^{t'_1,t'_2}$ of length $2^t$ have pairwise different
dimensions.  Hence, $(t'_1,t'_2)=(t_1,t_2)$, and the statement follows.
\end{proof}

The codes $H_2^{t_1,1}$, of length $2^{2t_1}$, are not covered by Theorem
\ref{thm:KV}.  In order to study them, we also need the rank of the
$\Z_8$-linear Hadamard codes.

\begin{theorem}[\cite{fernandez2019mathbb}]\label{thm:rank-z8}
Let $u_1\geq1$ and $u_2,u_3\geq0$.  Then,
\begin{equation}\label{eq:rank-z8}
\rank\bigl(\bH_2^{u_1,u_2,u_3}\bigr)=\frac{u_1^4-2u_1^3+35u_1^2+14u_1}{24}
+\frac{u_2}{2}\bigl(u_1^2+u_1+u_2+1\bigr)+u_3+1 .
\end{equation}
\end{theorem}

The following theorem is the main result of this section.

\begin{theorem}\label{thm:binary}
Let $t_1\geq2$ and $t_2\geq1$.  Then, the $\Z_2\Z_4$-linear Hadamard code
$H_2^{t_1,t_2}$ is equivalent to a $\Z_{2^s}$-linear Hadamard code of the
same length, for some $s\geq1$, if and only if $t_2\geq2$.
\end{theorem}

\begin{proof}
If $t_2\geq2$, then $H_2^{t_1,t_2}$ is equivalent to $\bH_2^{t_1+1,t_2-2}$, by
Theorem \ref{thm:KV}.  Conversely, let $H=H_2^{t_1,t_2}$, and let $2^t$ be
its length, where $t=2t_1+t_2-1$.  By Theorem \ref{thm:linZpZp2}, the code $H$
is nonlinear, since $t_1\geq2$.  Assume that $H$ is equivalent to a
$\Z_{2^s}$-linear Hadamard code of length $2^t$.  By Proposition
\ref{prop:candidates}, the code $H$ is equivalent to a nonlinear reduced code
$\bH_2^{u_1,\dots,u_r}$ of length $2^t$ whose derived code is linear and has
length $2^{t_1}$.  By Theorem \ref{thm:derivedZps} and Lemma
\ref{lem:invariance}, it follows that the code $\bH_2^{u_1,\dots,u_{r-1}}$ is
linear, and that its length is $2^{t_1}$.  By the discussion after Theorem
\ref{thm:derivedZps}, we have $r=2$, or $r=3$ and $u_1=2$.

Assume that $r=3$ and $u_1=2$.  By \eqref{eq:length}, the length of
$\bH_2^{2,u_2}$ is $2^{u_2+3}$.  Hence, $u_2=t_1-3$.  Moreover,
\eqref{eq:length} for $\bH_2^{2,u_2,u_3}$ gives $u_3=t_2$.  Then, by
\eqref{eq:rank-z8} and Theorem \ref{thm:rank}, we have
$$
\rank\bigl(\bH_2^{2,t_1-3,t_2}\bigr)=8+t_2+\frac{(t_1-3)(t_1+4)}{2}
=\rank\bigl(H_2^{t_1,t_2}\bigr)-(t_1-2).
$$
Since $t_1=u_2+3\geq3$, these two ranks are different, which contradicts
Lemma \ref{lem:invariance}.  Therefore, $r=2$.  By \eqref{eq:length}, the
code $\bH_2^{u_1}$ has length $2^{u_1-1}$, and hence $u_1=t_1+1$.  Since
$\bH_2^{u_1,u_2}$ has length $2^t$, \eqref{eq:length} gives
$2u_1+u_2=t+1=2t_1+t_2$, and hence $u_2=t_2-2$.  Since $u_2\geq0$, we obtain
$t_2\geq2$.
\end{proof}

Table \ref{tab:invariants-bin} compares the two codes $H_2^{t_1,t_2}$ and
$\bH_2^{2,t_1-3,t_2}$ which appear in the proof of Theorem \ref{thm:binary}.
They share all the invariants given by the kernel blocks.  The derived code of
$\bH_2^{2,t_1-3,t_2}$ is $\bH_2^{2,t_1-3}$, which is linear and has length
$2^{t_1}$.  Hence, it is equivalent to $\bH_2^{t_1+1}$, by Proposition
\ref{prop:reduction}.

\begin{table}[!ht]
\small
\renewcommand{\arraystretch}{1.3}
\centering
\caption{The invariants of $H_2^{t_1,t_2}$ and $\bH_2^{2,t_1-3,t_2}$, for
$t_1\geq3$.}
\label{tab:invariants-bin}
\vspace{1.5mm}
\begin{tabular}{|l|c|c|}
\hline
& $H_2^{t_1,t_2}$ & $\bH_2^{2,t_1-3,t_2}$\\
\hline\hline
length & $2^{2t_1+t_2-1}$ & $2^{2t_1+t_2-1}$\\ \hline
dimension of the kernel & $t_1+t_2$ & $t_1+t_2$\\ \hline
number of kernel blocks & $2^{t_1+t_2-1}$ & $2^{t_1+t_2-1}$\\ \hline
size of the kernel blocks & $2^{t_1}$ & $2^{t_1}$\\ \hline
derived code & $\bH_2^{t_1+1}$ & $\bH_2^{t_1+1}$\\ \hline
kernel-block rank profile & $t_1+1$, repeated $2^{t_1+t_2-1}$ times
& $t_1+1$, repeated $2^{t_1+t_2-1}$ times\\ \hline
rank & $2t_1+t_2+\binom{t_1}{2}$ & $2t_1+t_2+\binom{t_1}{2}-(t_1-2)$\\
\hline
\end{tabular}
\end{table}

For $t_2=1$, Theorem \ref{thm:binary} answers the question of
\cite{EquivZ2s} for every even $t\geq4$, as follows.

\begin{corollary}\label{cor:EquivZ2s}
Let $t\geq4$ be even.  Then, the $\Z_2\Z_4$-linear Hadamard code
$H_2^{t/2,1}$ is not equivalent to any $\Z_{2^s}$-linear Hadamard code of
length $2^t$, for any $s\geq1$.
\end{corollary}

\begin{proof}
Since $t/2\geq2$, the statement follows from Theorem \ref{thm:binary} with
$t_1=t/2$ and $t_2=1$.
\end{proof}

From the point of view of the $\Z_{2^s}$-linear Hadamard codes, we obtain the
following characterization.

\begin{corollary}\label{cor:Z4}
A nonlinear $\Z_{2^s}$-linear Hadamard code is equivalent to a code
$H_2^{t_1,t_2}$ if and only if it is equivalent to a $\Z_4$-linear Hadamard
code.
\end{corollary}

\begin{proof}
Let $\bH$ be a nonlinear $\Z_{2^s}$-linear Hadamard code.  Assume that $\bH$ is
equivalent to a code $H_2^{t_1,t_2}$.  By Lemma \ref{lem:invariance}, this code
is nonlinear, and hence $t_1\geq2$, by Theorem \ref{thm:linZpZp2}.  By Theorem
\ref{thm:binary}, we have $t_2\geq2$.  Hence, by Theorem \ref{thm:KV}, the code
$\bH$ is equivalent to the $\Z_4$-linear Hadamard code $\bH_2^{t_1+1,t_2-2}$.
Conversely, assume that $\bH$ is equivalent to a $\Z_4$-linear Hadamard code
$\bH_2^{u_1,u_2}$.  By Lemma \ref{lem:invariance}, this code is nonlinear, and
hence $u_1\geq3$, by Theorem \ref{thm:linZps}.  By Theorem \ref{thm:KV}, with
$t_1=u_1-1$ and $t_2=u_2+2$, the code $\bH_2^{u_1,u_2}$ is equivalent to
$H_2^{u_1-1,u_2+2}$.
\end{proof}

We now count the equivalence classes for $p=2$.

\begin{corollary}\label{cor:count-bin}
Let $t\geq3$.  Then, the codes $H_2^{t_1,t_2}$ of length $2^t$, together with
the $\Z_{2^s}$-linear Hadamard codes of length $2^t$ with $s\geq1$, form
exactly $\cP(t+1)-\cP(t)$ equivalence classes if $t$ is even, and
$\cP(t+1)-\cP(t)-1$ equivalence classes if $t$ is odd.  A complete system of
representatives is formed by the linear code $\bH_2^{t+1}$, the nonlinear
reduced codes of length $2^t$, and, if $t$ is even, the code $H_2^{t/2,1}$.
\end{corollary}

\begin{proof}
By Theorem \ref{thm:classZps}, the $\Z_{2^s}$-linear Hadamard codes of length
$2^t$ form $\cP(t+1)-\cP(t)-1$ equivalence classes, represented by
$\bH_2^{t+1}$ and the nonlinear reduced codes of length $2^t$.  By Corollary
\ref{cor:classZpZp2}, the codes $H_2^{t_1,t_2}$ of length $2^t$ are the codes
$H_2^{t_1,t+1-2t_1}$ with $1\leq t_1\leq\lfloor t/2\rfloor$.  If
$t+1-2t_1\geq2$, then $H_2^{t_1,t+1-2t_1}$ is equivalent to a
$\Z_4$-linear Hadamard code, by Theorem \ref{thm:KV}.  Hence, it lies in one
of the previous classes.  Otherwise, we have $t+1-2t_1=1$.  Then, $t$ is even,
and $t_1=t/2\geq2$.  By Corollary \ref{cor:EquivZ2s}, the code $H_2^{t/2,1}$
is not equivalent to any $\Z_{2^s}$-linear Hadamard code.  Therefore, it
gives one more class when $t$ is even, and the statement follows.
\end{proof}

Table \ref{tab:count} gives the numbers of equivalence classes of Corollaries
\ref{cor:count} and \ref{cor:count-bin} for $3\leq t\leq14$.  For odd $p$, the
row $\lfloor t/2\rfloor$ gives the number of classes of the codes
$H_p^{t_1,t_2}$, and the row $\cP(t+1)-\cP(t)$ gives the number of classes of
the $\Z_{p^s}$-linear GH codes.  These numbers do not depend on the odd prime
$p$.

\begin{table}[!ht]
\small
\setlength{\tabcolsep}{4.5pt}
\renewcommand{\arraystretch}{1.15}
\centering
\caption{The number of equivalence classes of Corollaries \ref{cor:count} and
\ref{cor:count-bin}.}
\label{tab:count}
\vspace{1.5mm}
\begin{tabular}{|c||c|c|c|c|c|c|c|c|c|c|c|c|}
\hline
$t$ & $3$ & $4$ & $5$ & $6$ & $7$ & $8$ & $9$ & $10$ & $11$ & $12$ & $13$
& $14$\\
\hline\hline
$\lfloor t/2\rfloor$ & 1 & 2 & 2 & 3 & 3 & 4 & 4 & 5 & 5 & 6 & 6 & 7\\
\hline
$\cP(t+1)-\cP(t)$ & 2 & 2 & 4 & 4 & 7 & 8 & 12 & 14 & 21 & 24 & 34 & 41\\
\hline\hline
total, $p$ odd & 3 & 4 & 6 & 7 & 10 & 12 & 16 & 19 & 26 & 30 & 40 & 48\\
\hline
total, $p=2$ & 1 & 2 & 3 & 4 & 6 & 8 & 11 & 14 & 20 & 24 & 33 & 41\\
\hline
\end{tabular}
\end{table}

\begin{example}\label{ex:binary}
Let $t=6$.  The codes $H_2^{t_1,t_2}$ of length $2^6$ are $H_2^{1,5}$,
$H_2^{2,3}$ and $H_2^{3,1}$.  By Theorem \ref{thm:KV}, the first two codes are
equivalent to the $\Z_4$-linear Hadamard codes $\bH_2^{2,3}$ and
$\bH_2^{3,1}$, respectively.  By Theorem \ref{thm:classZps}, the
$\Z_{2^s}$-linear Hadamard codes of length $2^6$ are represented by the linear
code $\bH_2^{7}$ and by the nonlinear reduced codes $\bH_2^{3,1}$ and
$\bH_2^{2,0,1}$.  The code $H_2^{3,1}$ has kernel dimension $4$.  Among the
nonlinear reduced codes, only $\bH_2^{2,0,1}$ has the same kernel dimension,
since $\bH_2^{3,1}$ has kernel dimension $5$.  By Theorems \ref{thm:blocks}
and \ref{thm:derived}, the code $H_2^{3,1}$ has $8$ kernel blocks of size
$8$, and its derived code is $\bH_2^{4}$.  By Theorem \ref{thm:derivedZps} and
Proposition \ref{prop:reduction}, the same holds for $\bH_2^{2,0,1}$.
However, the ranks of $H_2^{3,1}$ and $\bH_2^{2,0,1}$ are $10$ and $9$,
respectively, by Theorems \ref{thm:rank} and \ref{thm:rank-z8}.  Hence,
$H_2^{3,1}$ is not equivalent to any $\Z_{2^s}$-linear Hadamard code, and
there are $3+1=4$ equivalence classes.  This agrees with Corollary
\ref{cor:count-bin}, since $t$ is even and $\cP(7)-\cP(6)=15-11=4$.
\end{example}

Figure \ref{fig:classes} shows the equivalence classes of the codes
$H_p^{t_1,t_2}$ and of the $\Z_{p^s}$-linear GH codes of length $p^8$.  For odd
$p$, the two families have no common class, by Theorem \ref{thm:main}.  For
$p=2$, the linear class and the classes of the nonlinear $\Z_4$-linear Hadamard
codes are common to both families, by Theorem \ref{thm:KV} and Proposition
\ref{prop:reduction}.  The class of $H_2^{4,1}$ contains no $\Z_{2^s}$-linear
Hadamard code, by Corollary \ref{cor:EquivZ2s}.  Finally, by Theorem
\ref{thm:classZps} and Corollary \ref{cor:Z4}, the nonlinear reduced codes with
at least three parameters are not equivalent to any code $H_2^{t_1,t_2}$.

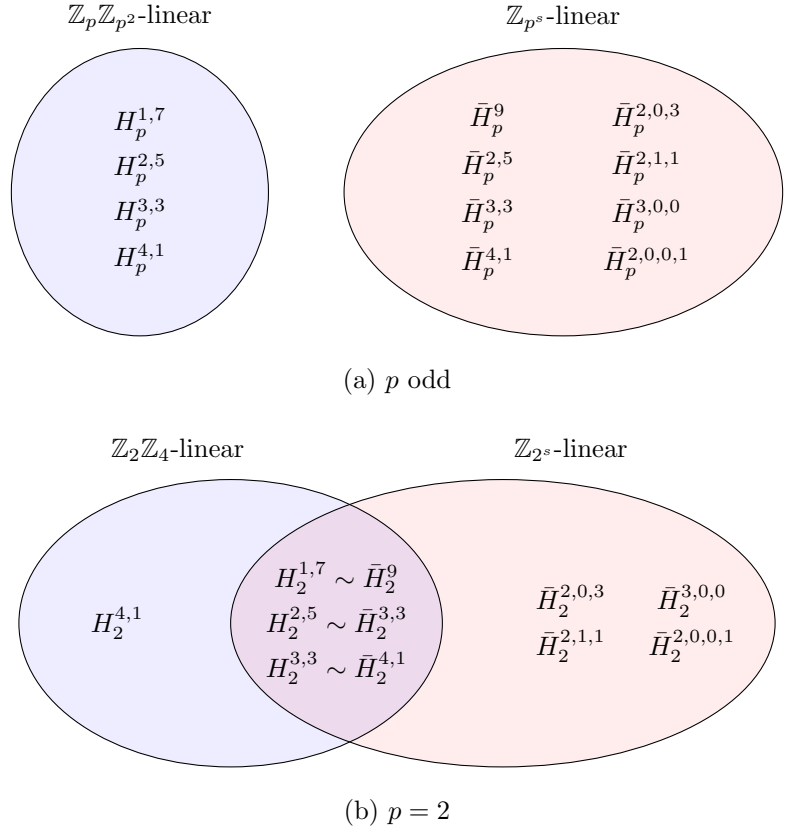
\begin{figure}[!ht]
\centering
\begin{tikzpicture}[font=\small]
\fill[blue!7] (0,0) ellipse (1.7cm and 1.9cm);
\fill[red!7] (5.6,0) ellipse (2.9cm and 1.9cm);
\draw (0,0) ellipse (1.7cm and 1.9cm);
\draw (5.6,0) ellipse (2.9cm and 1.9cm);
\node at (0,2.3) {$\Z_p\Z_{p^2}$-linear};
\node at (5.6,2.3) {$\Z_{p^s}$-linear};
\node[align=center] at (0,0)
  {$H_p^{1,7}$\\[5pt] $H_p^{2,5}$\\[5pt] $H_p^{3,3}$\\[5pt] $H_p^{4,1}$};
\node[align=center] at (4.6,0)
  {$\bH_p^{9}$\\[5pt] $\bH_p^{2,5}$\\[5pt] $\bH_p^{3,3}$\\[5pt] $\bH_p^{4,1}$};
\node[align=center] at (6.7,0)
  {$\bH_p^{2,0,3}$\\[5pt] $\bH_p^{2,1,1}$\\[5pt] $\bH_p^{3,0,0}$\\[5pt] $\bH_p^{2,0,0,1}$};
\node at (3.4,-2.5) {(a) $p$ odd};
\begin{scope}[xshift=1.2cm,yshift=-5.7cm]
\fill[blue!7] (0,0) ellipse (2.8cm and 1.9cm);
\fill[red!7] (3.6,0) ellipse (3.6cm and 1.9cm);
\begin{scope}
\clip (0,0) ellipse (2.8cm and 1.9cm);
\fill[violet!14] (3.6,0) ellipse (3.6cm and 1.9cm);
\end{scope}
\draw (0,0) ellipse (2.8cm and 1.9cm);
\draw (3.6,0) ellipse (3.6cm and 1.9cm);
\node at (-0.7,2.3) {$\Z_2\Z_4$-linear};
\node at (4.5,2.3) {$\Z_{2^s}$-linear};
\node at (-1.5,0) {$H_2^{4,1}$};
\node[align=center] at (1.4,0)
  {$H_2^{1,7}\sim\bH_2^{9}$\\[5pt] $H_2^{2,5}\sim\bH_2^{3,3}$\\[5pt]
   $H_2^{3,3}\sim\bH_2^{4,1}$};
\node[align=center] at (4.5,0) {$\bH_2^{2,0,3}$\\[5pt] $\bH_2^{2,1,1}$};
\node[align=center] at (6.1,0) {$\bH_2^{3,0,0}$\\[5pt] $\bH_2^{2,0,0,1}$};
\node at (2.2,-2.5) {(b) $p=2$};
\end{scope}
\end{tikzpicture}
\caption{The equivalence classes of the codes $H_p^{t_1,t_2}$ (left) and of
the $\Z_{p^s}$-linear GH codes (right) of length $p^8$.  Codes joined by
$\sim$ are equivalent.}
\label{fig:classes}
\end{figure}

\FloatBarrier

\section{Conclusions and further research}\label{sec:conclusions}

We have compared the $\Z_p\Z_{p^2}$-linear GH codes $H_p^{t_1,t_2}$ with the
$\Z_{p^s}$-linear GH codes, for every prime $p$.  For odd $p$, the codes
$H_p^{t_1,t_2}$ are not equivalent to any $\Z_{p^s}$-linear GH code.  This
answers, for every length, the question raised in
\cite{ZpZp2Construction,ZpZp2Classification}.  For $p=2$, a nonlinear code
$H_2^{t_1,t_2}$ is equivalent to a $\Z_{2^s}$-linear Hadamard code if and only
if $t_2\geq2$.  In particular, for every even $t\geq4$, the code $H_2^{t/2,1}$
is not equivalent to any $\Z_{2^s}$-linear Hadamard code, which answers the
question raised in \cite{EquivZ2s}.  As a consequence, we have obtained the
number of equivalence classes of the union of both families, for every prime
$p$.

In both cases, the proof combines two invariants.  First, the derived code
determines, up to equivalence, the $\Z_{p^s}$-linear GH codes which could be
equivalent to a nonlinear code $H_p^{t_1,t_2}$.  For odd $p$, the only such
code is $\bH_p^{t_1,t_2}$ if $t_1\geq2$, and there is none if $t_1=1$.  For
$p=2$, there are at most two such codes, namely $\bH_2^{t_1+1,t_2-2}$, when
$t_2\geq2$, and $\bH_2^{2,t_1-3,t_2}$, when $t_1\geq3$.  Second, the rank
separates these codes, except for the equivalent codes of Theorem \ref{thm:KV}.
For this purpose, we have determined the ranks of the codes $H_p^{t_1,t_2}$ and
of all the $\Z_{p^2}$-linear GH codes $\bH_p^{t_1,t_2}$, for every prime $p$.
For odd $p$, the rank of $H_p^{t_1,t_2}$ depends on the Bernoulli numbers
modulo $p$, while the rank of $\bH_p^{t_1,t_2}$ does not.

We end with two directions for further research.  First, for $p=2$, every
Hadamard code which is the Gray map image of a $\Z_{2^s}$-additive code is
equivalent to a code $\bH_2^{t_1,\dots,t_s}$ \cite{KernelZ2s,Krotov:2007}.
Moreover, every $\Z_2\Z_4$-linear Hadamard code of length $2^t$, with $t\geq3$
and $\alpha_1\neq0$, is equivalent to a code $H_2^{t_1,t_2}$, by \cite{PRV06}
and Corollary \ref{cor:classZpZp2}.  For odd $p$, it is not known whether every
GH code which is the Gray map image of a $\Z_{p^s}$-additive code is equivalent
to a code $\bH_p^{t_1,\dots,t_s}$.  It is also not known whether every
nonlinear GH code which is the Gray map image of a $\Z_p\Z_{p^2}$-additive code
with $\alpha_1\neq0$ is equivalent to a code $H_p^{t_1,t_2}$.  Positive answers
would show that, for odd $p$, no such nonlinear code is equivalent to the Gray
map image of a $\Z_{p^s}$-additive code.  Second, the ranks of the
$\Z_{p^s}$-linear GH codes with $s\geq3$ are not known for odd $p$.  Our method
describes the digit of index $1$ of the codewords by polynomial functions over
$\Z_p$.  It would be interesting to extend it to the digits of higher index,
and to see whether the Bernoulli numbers modulo $p$ appear again.

{\small
\bibliographystyle{abbrv}
\bibliography{ref}
}

\end{document}